\documentclass[sn-mathphys-num]{sn-jnl}
\usepackage{multirow}
\usepackage{amsmath,amssymb,amsfonts}
\usepackage{amsthm}
\usepackage{mathrsfs}
\usepackage[title]{appendix}
\usepackage{xcolor}
\usepackage{textcomp}
\usepackage{enumitem}
\usepackage{manyfoot}
\usepackage{bm}
\usepackage{listings}%
\theoremstyle{plain}
\usepackage{algorithm}
\usepackage{algorithmic}
\usepackage{array}
\usepackage[algo2e]{algorithm2e}
\usepackage[caption=false,font=normalsize,labelfont=sf,textfont=sf]{subfig}
\usepackage{booktabs}
\usepackage{cite}
\usepackage{multirow}
\usepackage{graphicx} 
\usepackage{hyperref}
\usepackage{stfloats}
\usepackage{url}
\usepackage{verbatim}
\usepackage{hyperref}
\theoremstyle{thmstyleone}%
\newtheorem{theorem}{Theorem}
\newtheorem{proposition}{Proposition}
\newtheorem{assumption}{Assumption}
\newtheorem{corollary}{Corollary}
\theoremstyle{thmstyletwo}%

\newtheorem{remark}{Remark}%
\theoremstyle{thmstylethree}%
\usepackage[title]{appendix}
\usepackage{caption}
\usepackage{cite}
\usepackage{subcaption}

\newcommand{\V}{\mathbf{V}}
\newcommand{\X}{\mathbf{X}}
\newcommand{\A}{\mathbf{A}}

\newcommand{\I}{\mathbf{I}}
\newcommand{\xt}{\tilde{\mathbf{X}}}
\newcommand{\Vk}{\mathbf{V}_{K^*}}
\newcommand{\calS}{\mathcal{S}}

\newcommand{\bbE}{\mathbb{E}}
\newcommand{\bbR}{\mathbb{R}}
\newcommand{\OVS}{\mathrm{OVS}}

\newcommand{\abs}[1]{\left|#1\right|}
\DeclareMathOperator*{\argmax}{arg\,max}

\DeclareMathOperator{\sign}{sign}

\begin{document}
\title{CaSPECT: Discovering Causally Homogeneous Subgroups via Directed Spectral Clustering}

\author[1]{\fnm{Arghya} \sur{Pratihar}}\email{arghyapratihar24@gmail.com}

\author[2]{\fnm{Shinjon} \sur{Chakraborty}}\email{shinjonchakraborty07@gmail.com}
\author*[1]{\fnm{Swagatam} \sur{Das}}\email{swagatam.das@isical.ac.in}

\affil*[1]{\orgdiv{Electronics and Communication Sciences Unit}, \orgname{Indian Statistical Institute}, \orgaddress{\country{India}}}

\affil[2]{\orgname{Indian Statistical Institute}, \orgaddress{\country{India}}}


\abstract{
We propose \textbf{CaSPECT}, a causal spectral clustering framework for discovering causally homogeneous subgroups from observational data. Rather than clustering in covariate space, CaSPECT defines similarity through the topology of a learned directed acyclic graph (DAG); a bootstrap-stabilised PC algorithm recovers the causal skeleton; a novel \emph{Orientation Validation Score} (OVS) combines PC bootstrap evidence with DirectLiNGAM to orient edges robustly; directed edges are weighted by backdoor-identified average treatment effects estimated via OLS or double machine learning. Chung's directed Laplacian provides a spectral embedding in which individuals close together share the same causal propagation pathways. We establish almost-sure consistency of the full pipeline and validate the method through a controlled simulation study and on LaLonde CPS1, IHDP, and 401(k) datasets, where CaSPECT recovers a positive and statistically significant treatment effect within the causally comparable subpopulation and corrects for severe confounding without requiring a pre-specified propensity score model.}

\maketitle

\section{Introduction}\label{intro}
Statistical Causal Inference aims to estimate how an outcome behaves in the presence of interventions on a causal variable. There are a few well-established estimates in this domain for quantifying causality, such as the Average Causal Effect (ACE), Conditional Average Causal Effect (CACE), Individual Causal Effect (ICE), etc. Among these, the average causal effect for a binary treatment $\{0,1\} \in A$ is defined as $E[Y_i(1)-Y_i(0)],$ where $Y_i(z)$ denotes the potential outcome under the treatment level $Z_i=z$. For observational data, much progress has been made in the causal literature in obtaining a consistent and robust estimate of the average causal effect, but it is well recognized that ACE may obscure the estimation of treatment effect heterogeneity. This heterogeneity, i.e., variability in treatment effects across individuals or subpopulations, is particularly relevant in domains such as healthcare, policy evaluation, etc., where interventions may produce beneficial results in some subgroups while harming others.       

One of the possible ways that has been given considerable attention to address this heterogeneity is the conditional average causal effect (CACE), defined as 
$
    \psi(X)=E[Y_i(1) -Y_i(0)|X],
$ 
where $X$ denotes the vector of observed covariates. The CACE facilitates personalized effect estimation and has motivated numerous methodological advances, particularly those including machine learning. Most of the interest while addressing the heterogeneity problem in the past has been directed towards supervised approaches. In contrast, recently, some limelight has been directed towards addressing the problem through unsupervised approaches like Clustering. So the objective of causal clustering is to uncover latent subgroup structures by clustering individuals based on estimated causal effect functions. This descriptive and data-driven perspective offers a complementary direction to existing approaches and fills a critical gap in the literature on causal inference. An instance of this causal clustering algorithm is addressed by \citep{hierarchical}, \citep{causalkmeans}, where they mainly focused on clustering with $k$-means, density-based clustering, and hierarchical clustering. The methodological challenge in this setting stems from the fact that the object being clustered, i.e., the causal effects, is not directly observed but estimated. This distinguishes the problem from the classical clustering on observed features. In light of the largely unexplored interface between spectral clustering and causal inference, we have proposed Causal Spectral Clustering (CaSPECT). This approach seeks to address the heterogeneity of causal effects, considering the geometry of the data through a similarity graph constructed over estimated treatment effect functions. The motivation for this approach stems from the fact that causal systems exhibit an inherent directed acyclic graph (DAG) structure according to the probabilistic graphical model framework, which implies intrinsic relationships among the variables that can be captured through spectral embeddings.

\section{Related methods}
Recent advances involving causal inference and unsupervised learning have introduced the notion of causal clustering, where the goal is to partition units according to heterogeneous treatment effects rather than observed covariate similarity. Two recent contributions in this direction are causal $k$-means clustering \citep{causalkmeans}, hierarchical and density-based causal clustering \citep{hierarchical}. Causal $k$-means clustering reformulates the classical $k$-means objective in terms of latent counterfactual response functions. Instead of clustering on observed features, units are grouped based on similarity in their individual treatment effects. Since these effects are not directly observed, the method relies on plug-in and doubly robust estimators. The authors establish consistency and asymptotic normality of the resulting clustering procedure, highlighting the role of semiparametric efficiency and nuisance function estimation. Extending this idea, \citep{hierarchical} generalizes causal clustering to hierarchical and density-based methods. This allows clusters to be defined through connectivity or density in the space of estimated causal effects, rather than relying on centroid-based partitions. Such flexibility is important when treatment effect heterogeneity exhibits a nonconvex or manifold structure. Theoretical guarantees are derived showing that the statistical performance depends critically on the estimation error of nuisance components, such as outcome regressions and propensity scores.

These approaches differ fundamentally from classical clustering techniques such as spectral clustering \citep{spectral}, which operate on similarity graphs constructed from observed data. Standard spectral methods involve the eigenstructure of the graph Laplacian to identify clusters with strong ingroup similarity, but are agnostic to causal structure. As a result, they may fail to recover groups with homogeneous treatment effects when such a structure is not aligned with covariate similarity. The recent causal clustering literature, therefore, motivates the use of causally informed similarity measures, where affinities are constructed using estimated treatment effects or pseudo-outcomes. A causal spectral clustering framework naturally combines the flexibility of graph-based methods with semiparametric causal estimation. Compared to $k$-means \citep{kmeans} based approaches, spectral methods can capture nonlinear and graph-structured heterogeneity while accommodating doubly robust estimation of nuisance functions.

In summary, existing work highlights three key insights: (i) clustering should be performed on counterfactual quantities, (ii) statistical efficiency requires careful handling of nuisance estimation, and (iii) flexible clustering paradigms are needed to capture complex heterogeneity. These ideas directly motivate the development of causal spectral clustering.

\section{Our Contributions}
The proposed Causal Spectral Clustering (CaSPECT) framework changes the clustering approach from simple feature similarity to causal topology. Instead of grouping individuals by observed covariates or differences in treatment effect, CaSPECT reveals the deeper structural and causal relationships that drive the data-generating process. The framework starts by recovering a Directed Acyclic Graph (DAG) using the PC algorithm \citep{pc_algo}. It then assigns weights to the edges through various estimation methods and creates a transition operator to model how causal influence spreads. By performing spectral decomposition on a symmetrized Laplacian operator, our method produces a low-dimensional representation that captures the overall shape of this causal diffusion landscape. This method groups individuals based on their active pathways rather than their proximity in Euclidean space. This structural approach allows CaSPECT to uncover latent causal regimes that traditional methods frequently overlook. Unlike mixture Structural Equation Models (mixture SEMs), which fit specific models based on groups before clustering on parameters, our method reverses this sequence in turn, allowing interpretation of the graph architecture to cluster individuals without priorly specifying parametric forms. Furthermore, because it does not cluster in covariate space, CaSPECT successfully distinguishes overlapping subpopulations where Conditional Average Causal Effect (CACE) methods and mixture SEMs fail due to marginal distribution overlap. It doesn't ask what values variables take, but how influence flows, providing a principled framework for mapping specific treatment responses for particular subpopulations, directly to localized structural interventions. The primary trade-off of this methodology is its dependency on the accuracy of the initial DAG and the assumption of causal sufficiency inherent to the standard PC algorithm. While stability weighting based on bootstrapping can downweight inconsistent edges and mitigate minor uncertainty, a fundamentally misspecified graph will propagate errors throughout the entire spectral embedding. Consequently, CaSPECT trades graph-independent robustness for unparalleled resolution of structural heterogeneity and targeted policy insights. So, the higher the accuracy of the graph formation mechanism, the better this algorithm performs.

\section{Notations \& Preliminaries} \label{sec:prelims}
The observed dataset is arranged as a matrix
$\X \in \bbR^{n \times q}$, where $n$ is the number of observations and $q$
is the number of variables. The variable set
$\V = \{X_1, \ldots, X_q, Z, Y\}$ consists of background covariates, a
treatment variable $Z$, and an outcome $Y$.
 
Causal relationships among the variables in $\V$ are encoded in a directed acyclic graph. We write $G_0$ for the true DAG that generated the data and $E_0$ for its skeleton, the underlying undirected edge set. A directed edge from $u$ to $v$ is written $u \to v$, and the backdoor adjustment set required to identify the causal effect of $u$ on $v$ is denoted $\X_{uv}$. The average causal effect along that edge is $\tau_{u \to v}$, and the corresponding structural coefficient is $\beta_{uv}$. Once the causal graph is estimated and any ambiguous edges are resolved via contraction (sec \ref{sec:method}, step 1), the resulting contracted variable set ${\bf V}^*$. The directed edges among $\V^*$ are collected into an adjacency matrix $\A \in \bbR^{|{\bf V}|^* \times |{\bf V}^*|}$, with entry $\A_{uv}$ denoting the estimated causal effect magnitude for $u \to v$. From $\A$ we construct a row-stochastic transition matrix $P$, whose stationary distribution is $\pi$. These feed into the Chung Laplacian $L$, whose eigenvalues $\lambda_1 \leq \lambda_2 \leq \cdots$ and eigenvectors $\mathbf{v}_1, \mathbf{v}_2, \ldots$ give the spectral embedding. The embedding of the full dataset is $\xt \in \bbR^{n \times K^*}$, where $K^*$ is the chosen embedding dimension, and $\tilde{\mathbf{x}}_i$ denotes the embedding of individual observation $i$. We work within the potential outcomes framework for causal inference \citep{Neyman:1923, Rubin:1974} and assume a sample of $n$ units to be randomly selected from a population. For each unit $i = 1, \ldots, n$, $Z_i$ is an assignment indicator to treatment with $Z_i = 1$ if assigned to treatment and $ Z_i = 0$ otherwise. Let $X_i \in \mathbb{R}^q$ be a vector of observed covariates, and $Y_i(Z_i)$ is the observed outcome variable depending on the assignment of the treatment variables. Let $\{Y_i(1), Y_i(0)\}$ be the potential outcomes when treated and not treated, respectively, with only one of them observed in the sample, also referred to as counterfactual outcomes.

\section{Assumptions}
\label{sec:assumptions}
The following assumptions underpin the CaSPECT pipeline and are grouped according to the stage at which they are invoked.

\begin{assumption}[Faithfulness]\label{ass:B1}
The joint distribution $P(\V)$ is faithful to the true DAG $G_0$. Every conditional independence in $P(\V)$ corresponds to a $d$-separation in $G_0$, and every $d$-separation in $G_0$ corresponds to a conditional independence in $P(\V)$ \citep{spirtes2000causation}.
\end{assumption}

\begin{assumption}[Causal Sufficiency]\label{ass:B2}
The observed variable set $\mathbf{V}$ is assumed to be causally sufficient; that is, there exists no unobserved variable that acts as a common cause of any pair of variables in $\mathbf{V}$. This assumption is required for the consistency guarantees of the PC algorithm \citep{pc_algo}.  
\end{assumption}
 
\begin{assumption}[Non Gaussianity]\label{ass:B3}
At least half the variables in $\V$ have non-Gaussian error distributions \citep{shimizu2006linear}. When this fails, LiNGAM's weight in the OVS is scaled down proportionally (see Section~\ref{sec:dag}).
\end{assumption}
 
\begin{assumption}[Per edge Linearity or Identifiable Nonlinearity]
\label{ass:B4}
Each edge $(u \to v)$ satisfies either the linear SEM
$v = \beta_{uv}u + \bm{\beta}_X^\top\X_{uv} + \varepsilon_v$, confirmed
by a RESET test, or the partially linear model of
\citep{chernozhukov2018double}. Prior CaSPECT formulations assume global
linearity; we relax this to the edge level.
\end{assumption}
 
\begin{assumption}[Causal Identification]\label{ass:B5}
For every edge $(u \to v)$, consistency, conditional ignorability given
the backdoor set $\X_{uv}$, and positivity hold
\citep{pearl,Neyman:1923,Rubin:1974}. These three assumptions, viz. consistency, ignorability, and positivity, clubbed together as causal identification, are the standard identification requirements for any ACE estimate based on backdoor adjustment.
\end{assumption}

Assumptions~\ref{ass:B1}--\ref{ass:B3} are inherited from the component algorithms, the PC algorithm, and LiNGAM and are well understood in the causal discovery literature. Assumptions~\ref{ass:B4}--\ref{ass:B5} are specific to our edge weight estimation step. 

\section{Proposed Method}
\label{sec:method}
We propose CaSPECT, a causal spectral clustering framework that takes observational data  $\mathbf{X} \in \mathbb{R}^{n \times q}$ over variables $\mathbf{V} = \{X_1, \ldots, X_q, Z, Y\}$ and produces cluster assignments grounded in the underlying causal structure through four sequential stages: DAG formation, causal edge weight estimation, Chung Laplacian construction, spectral embedding and clustering. We discuss each step of our proposed method, CaSPECT, in detail. A concise summary of our algorithm is presented in Algorithm \ref{alg:Caspect}.

\vspace{3pt}
\noindent
\textbf{Step 1: DAG Formation.}
\label{sec:dag}
We construct a fully oriented DAG $G^* = (\V^*, E^*)$ in three stages. In the first stage, we use the PC Algorithm with bootstrap filtering to obtain a stable causal skeleton. Then we introduce the OVS to resolve orientation ambiguities left by PC alone. Finally, we handle any edges that remain undirected after the OVS through some hierarchical sequence of procedures.
 
We apply the PC Algorithm \citep{pc_algo,spirtes2000causation} at
significance level $\alpha_{\mathrm{CI}}$ to obtain a CPDAG
$G_{\mathrm{PC}}$ over $\V$. Orientation of v-structures and subsequent propagation follows Meek's rules \citep{meek2013causal}. Under Assumptions~\ref{ass:B1} and~\ref{ass:B2}, the algorithm recovers the true CPDAG as $n \to \infty$ \citep{spirtes2000causation}. From computational point of view, running PC on a single dataset conflates genuine causal edges with sampling artifacts. To fix this ambiguity, we re run PC on $B$ bootstrap resamples and record, for each edge $(u,v)$, the inclusion frequency $f_{uv}$ and the orientation frequency $g_{uv}$ (fraction of resamples in which $u \to v$ is returned). Only edges with $f_{uv} \geq \theta$ enter the stable skeleton $\calS$, which bounds $\bbE[\mathrm{False Positives(FP)}] \leq \frac{q^2}{2B}$ by stability selection \citep{meinshausen2010stability}. The ratio
$\rho_{uv} = g_{uv}/f_{uv}$ carries per-edge orientation confidence forward into the OVS.
 
 
The PC Algorithm cannot orient all edges from observational data because Markov-equivalent DAGs share the same conditional independence structure. LiNGAM \citep{shimizu2006linear,shimizu2011directlingam} exploits a different property of the data, non-Gaussianity of the error terms, to resolve this ambiguity. Because the two methods rest on independent assumptions, disagreement between them is a genuine signal rather than redundancy. Before using LiNGAM evidence, we apply a per-variable Jarque-Bera test \citep{jarque} to OLS residuals. The LiNGAM weight is set to $w_L =
(n_{\mathrm{NG}}/q) \cdot w_{max}$, where $n_{\mathrm{NG}}$ is the number of variables whose residuals reject Gaussianity at a pre specified significance level $\alpha_{s}$, $q$ is the total number of variables and $w_{max} \in (0, 0.5)$ is a user defined hyperparameter bounding the maximum permissible influence of the LiNGAM term. 
The OVS for each edge in $\calS$ is:
\begin{equation}\label{eq:ovs}
   \OVS_{uv}
   = w_{\mathrm{PC}} \cdot f_{uv}
     \cdot \underbrace{(\rho_{uv} - \rho_{vu})}_{\delta_{uv}^{\mathrm{PC}}}
   + w_L \cdot
     \underbrace{\sign(\hat{B}_{uv} - \hat{B}_{vu})}_{\delta_{uv}^{L}},
\end{equation}
where $w_{\mathrm{PC}} + w_L = 1$. The term $\delta_{uv}^{\mathrm{PC}}
\in [-1,+1]$ encodes net bootstrap orientation evidence, weighted by
$f_{uv}$ to reflect edge stability. \begin{equation*}
    \delta_{uv}^{L}
    \;=\;
    \sign(\hat{B}_{uv} - \hat{B}_{vu})
    \;\in\; \{-1,\, 0,\, +1\},
\end{equation*} 
where $+1$ indicates that LiNGAM prefers $u \to v$, $-1$ indicates
preference for $v \to u$, and $0$ indicates exact indifference which is
a measure-zero event under continuous distributions, and it rarely occurs in
practice; $\hat{B}_{uv}$ is the LiNGAM estimated coefficient from u to v, it represents how strongly LiNGAM believes u causally drives v in a linear non-Gaussian structural equation model. Let $\gamma \in (0, 1)$ be the edge orientation threshold. An edge is oriented $u \to v$ if $\OVS_{uv} > \gamma$, as $v \to u$ if
$\OVS_{uv} < -\gamma$, and passed to the next step if $\abs{\OVS_{uv}} \leq
\gamma$ (we have considered $\gamma = 0.15$ as default during the computations).
 
\begin{proposition}[OVS Consistency]\label{prop:ovs}
 Under Assumptions \ref{ass:B1}- \ref{ass:B3}, if $w_L > \gamma$, then $\text{OVS}_{uv}$ asymptotically strictly bounds away from $[-\gamma, \gamma]$ with the correct sign for all true edges as $n \to \infty$. Consequently, the Orientation Validity Score consistently orients the true DAG in large samples.
\end{proposition}

\begin{proof}
    The detailed proof is given in Appendix \ref{prop:proof}.
\end{proof}

\begin{algorithm}[t]
\caption{CaSPECT Algorithm}
\label{alg:Caspect}
\KwIn{$\X \in \mathbb{R}^{n \times q}$, Z, Y, $\alpha$, $B$, $\theta$, $\gamma$}
\KwOut{Clusters $\mathcal{C}$}

\textbf{DAG Formation:} \\
Run PC Algorithm $\Rightarrow G_{\mathrm{PC}}$; \\
Bootstrap $B$ times $\Rightarrow (f_{uv}, g_{uv})$; $\mathcal{S}=\{f_{uv}\ge\theta\}$; \\
Run LiNGAM and set weight $w_L$; \\
Compute
\[
\OVS_{uv}=w_{\mathrm{PC}} f_{uv}(\rho_{uv}-\rho_{vu}) + w_L\,\mathrm{sign}(\hat{B}_{uv}-\hat{B}_{vu});
\]
Orient via $\gamma$; resolve remaining edges (acyclicity, propagation rules, contraction); \\
Verify DAG via acyclicity constraint.

\textbf{Edge Weights:} \\
$A_{uv} \leftarrow |\hat{\beta}_{uv}|$ OLS or DML estimate; \\
$A^{\mathrm{stab}}_{uv} \leftarrow f_{uv} A_{uv}$.

\textbf{Laplacian:} \\
Construct $P$, $L=I-\tfrac{1}{2}(P+P^*)$; \\
$K^* \leftarrow \arg\max_j(\lambda_{j+1}-\lambda_j)$.

\textbf{Clustering:} \\
$\tilde{\X}=\X^*\Vk$; apply $k$-means clustering.

\Return $\mathcal{C}$
\end{algorithm}
The edges with $|\OVS_{uv}| \leq \gamma$ are resolved by the following hierarchy. Acyclicity is checked after each step.
 
\begin{enumerate}
  \item \textbf{Acyclicity filter.} If one orientation would create a directed cycle and the other would not, the cycle-free orientation is assigned.
 
  \item \textbf{Meek rules} \citep{meek2013causal}\textbf{.} All four propagation rules are applied exhaustively to the partial DAG.
 
  \item \textbf{Domain or temporal ordering.} Where prior knowledge
  establishes that $u$ precedes $v$, orient $u \to v$.
 
  \item \textbf{Edge contraction.} Any edge still undirected after
   Steps~1--3 is handled by merging $u$ and $v$ into a single node
   $N_{uv}$ with averaged incident weights: $A_{N \to w} = (A_{uw} +
   A_{vw})/2$ and $A_{w \to N} = (A_{wu} + A_{wv})/2$.
\end{enumerate}
 
\begin{remark} \label{rem2}
Contraction is preferable to the common alternative of symmetrizing
($A_{uv} \leftarrow A_{vu}$), which creates the cycle $u \to v \to u$.
Such a cycle renders the backdoor criterion \citep{pearl} inapplicable, prevents ACE identification, and invalidates the causal reading of the spectral embedding.
\end{remark}
 
\noindent
Global acyclicity of $G^*$ is confirmed via the NOTEARS condition \citep{zheng2018dags}. If any cycle persists, the edge in that cycle with the lowest $|\OVS|$ value is contracted.

\vspace{5pt}
\noindent
\textbf{Step 2: Causal Edge Weight Estimation.}
\label{sec:weights}

Each directed edge $(u \to v) \in E^*$ is weighted by the magnitude of ACE of $u$ on $v$, identified via the backdoor adjustment set $\X_{uv}$ read off from $G^*$ \citep{pearl}. Then we perform a RESET test on each edge to select the estimation approach. The two approaches are explained below.

\noindent 
\textbf{OLS for Linear Edges.}
Where the RESET test does not reject linearity ($p > 0.05$), OLS of $v$ on $u$ and $\X_{uv}$ gives a consistent estimate of $\tau_{u \to v} =
\beta_{uv}$ under Assumptions~\ref{ass:B4} and~\ref{ass:B5}, so we set $A_{uv} = |\hat{\beta}_{uv}|$.
 
\noindent 
\textbf{Double Machine Learning for Nonlinear Edges.}
Where linearity is rejected, we use DML \citep{chernozhukov2018double} with $D=5$-fold cross-fitting. Flexible nuisance estimators (GAM or Random Forest) \citep{gam, random} partial out the confounding effect of $\X_{uv}$ on both $u$ and $v$, providing cross-fitted residuals $\tilde{U}_i$ and
$\tilde{V}_i$. The ACE estimate is:
\begin{equation}\label{eq:dml}
   \hat{\tau}_{\mathrm{DML}}
   = \frac{\sum_i \tilde{U}_i\tilde{V}_i}{\sum_i \tilde{U}_i^2},
   \qquad A_{uv} = \abs{\hat{\tau}_{\mathrm{DML}}}.
\end{equation}
This estimator is $\sqrt{n}$-consistent and asymptotically normal
regardless of the rate at which the nuisance functions are estimated.

The OLS and DML estimators provide estimates of the causal effect magnitude $A_{uv}$ for each directed edge $u \to v$. However, these estimates do not account for uncertainty in the recovered graph structure $G^*$. To incorporate structural uncertainty, we weight each estimated effect by its bootstrap inclusion frequency. Specifically, the final adjacency entry is defined as
$A^{\mathrm{stab}}_{uv}=f_{uv}A_{uv},$ where $f_{uv}\in[0,1]$ denotes the proportion of bootstrap samples in which the edge $u\to v$ appears. This weighting scheme attenuates the contribution of edges that are estimated inconsistently across bootstrap replications while preserving the effect magnitudes of highly stable edges. Consequently, the resulting adjacency matrix
$\mathbf{A}^{\mathrm{stab}}= \bigl(A^{\mathrm{stab}}_{uv}\bigr)$
encodes both causal effect strength and structural confidence. Rather than imposing a hard threshold and discarding uncertain edges entirely, bootstrap stability weighting provides a continuous adjustment that propagates graph uncertainty into subsequent stages of the algorithm.

\begin{theorem}
\label{thm:ate}
Let $G^* = (\mathbf{V}^*, E^*)$ be the fully oriented DAG produced by the DAG formation stage of CaSPECT, and let $(u \to v) \in E^*$ be a directed edge with backdoor adjustment set $\mathbf{X}_{uv}$ identified from $G^*$. Assume consistency, conditional ignorability given $\mathbf{X}_{uv}$, and positivity hold for the intervention on $u$. Then the bootstrap stability weighted adjacency entry
$
A_{uv}^{\mathrm{stab}} = f_{uv} \cdot A_{uv}
$
satisfies
$
A_{uv}^{\mathrm{stab}} \xrightarrow{a.s.} |\tau_{u \to v}^0|
$
as $n \to \infty$, where $\tau_{u \to v}^0$ is the true average causal effect (ACE) of $u$ on $v$.
\end{theorem}

\begin{proof}
    The detailed proof is given in the Appendix \ref{thm:proof}.
\end{proof}

\begin{corollary}
\label{cor1}
Under the conditions of the theorem,
$
A_{uv}^{\mathrm{stab}}
$
consistently estimates
$
|\tau_{u \to v}^0|
$
for every true directed edge
$
(u \to v) \in E^*,
$
and converges to zero for every false edge. Edges with low bootstrap inclusion frequency
$
f_{uv}
$
are proportionally shrunk toward zero, reflecting genuine structural uncertainty, without invalidating identification for edges whose presence is stable across resamples.
\end{corollary}
Corollary \ref{cor1} gives $\mathbf {A}^{\mathrm{stab}}$ a precise structural interpretation, which can be described in brief as the entries corresponding to the stable causal edges converge to the true ACE magnitudes, whereas the entries corresponding to spurious edges decay asymptotically to zero without requiring any arbitrarily chosen threshold. This asymptotic separation is what enables the subsequent transition to spectral methods.
 Normalization of the rows $\mathbf{A}^{\mathrm{stab}}$ produces a transition matrix whose stationary distribution reflects the propagation of long-term causal influence through the graph. Chung's directed Laplacian, constructed from this transition matrix together with its time reversal, is symmetric and positive semi-definite precisely because the underlying edge weights are asymptotically well behaved. In contrast, noisy or spuriously inflated edge weights would distort the transition operator and violate the spectral regularity conditions required for stable eigen space recovery.

\vspace{5pt}
\noindent
\textbf{Step 3: Chung Laplacian and Spectral Decomposition.}
\label{sec:laplacian}

We normalize $\A^{\mathrm{stab}}$ with PageRank teleportation parameter $\alpha$ to obtain the transition matrix $P$ \citep{page1999pagerank}.
Let $D = \mathrm{diag}(\mathbf{A}^{\mathrm{stab}}\mathbf{1})$
be the diagonal matrix of the sums of the rows of $\mathbf{A}^{\mathrm{stab}}$, where
$\mathbf{1} \in \mathbb{R}^{|\mathbf{V}^*|}$ is the vector of all ones. The transition matrix $P$ is defined as,
\begin{equation}\label{eq:pagerank-matrix}
   P = (1-\alpha)\,D^{-1}\mathbf{A}^{\mathrm{stab}}
     + \frac{\alpha}{|\mathbf{V}^*|}\,\mathbf{1}\mathbf{1}^{\top},
\end{equation}
Its stationary distribution $\pi$ and the time-reversed matrix $P^*_{ij} = (\pi_j/\pi_i)\,P_{ji}$ yield Chung's directed Laplacian
\citep{chung2005laplacians}:
\begin{equation}\label{eq:laplacian}
   L = \I - \tfrac{1}{2}(P + P^*).
\end{equation}
$L$ is symmetric and positive semi-definite \citep{chung1997spectral}. We select the embedding dimension as: $ K^* = \argmax_{j}\,(\lambda_{j+1} - \lambda_j), \hspace{0.1cm} j \in \{1,2,  \ldots, (|\mathbf{V}^*|-1)\}$ taking the index of the largest spectral gap. The motivation comes from Cheeger's inequality \citep{chung2005laplacians}, i.e., a large gap at position $K^*$ corresponds to $K^*$ well-separated causal communities. The perturbation matrix $E=\hat{L}-L_0$ is the difference between the estimated Laplacian and the population Laplacian. The perturbation arises from two sources: the DAG estimation error in $\hat{G}^*$ and the ACE estimation error in ${\mathbf{A}}^{\mathrm{stab}}$. Define the sequence of events
$\mathcal{A}_n = \left\{\hat{G}^{*}=G_0
\right\}.$ On the event $\mathcal{A}_n$, the DAG estimation error vanishes, so \(E\) is driven entirely by ACE estimation error. For Track OLS, ordinary least squares estimation yields $|\hat{\beta}_{uv}-\beta_{uv}|
=
\mathcal{O}(n^{-1/2}),
$
by standard asymptotic theory. Track DML gives
$
|\hat{\tau}_{\mathrm{DML}}-\tau_{u\rightarrow v}|
=
\mathcal{O}(n^{-1/2}),
$
under the rate conditions of Theorem~3.1 in \citep{chernozhukov2018double}. Since \(\mathbf{A}^{\mathrm{stab}}\) contains at most $|\mathbf{V}^*|^2$ entries, and each entry converges at rate \(\mathcal{O}(n^{-1/2})\), the Frobenius norm of the adjacency estimation error satisfies $\|{\mathbf{A}}^{\mathrm{stab}}-\mathbf{A}_0\|_F =\mathcal{O}\!\left(|\mathbf{V}^*|\,n^{-1/2}\right)$, where $\mathbf{A}_0$ denotes the true population adjacency matrix weighted with ACE, whose entries $|\tau^0_{u \to v}|$ encode the magnitude of the true average causal effect along each directed edge $(u \to v) \in E^*$. Since the Laplacian construction is a continuous composition of bounded linear operators applied to \(\mathbf{A}^{\mathrm{stab}}\), the perturbation norm obeys
$
\|E\|_F
=
\mathcal{O}\!\left(|\mathbf{V}^*|\,n^{-1/2}\right).
$
$\mathbf{V}_{K^*}$ denotes the matrix of the $K^*$ non-trivial eigenvectors of the population Chung directed Laplacian $L_0$, i.e. $\mathbf{V}_{K^*} = [\mathbf{v}_2, \mathbf{v}_3, \ldots, \mathbf{v}_{K^*+1}]$. The first eigenvector $\mathbf{v}_1$ associated with $\lambda_1 = 0$ is the trivial constant eigenvector and is excluded because it carries no structural information about causal community separation. Its estimated counterpart $\hat{\mathbf{V}}_{K^*}$ is the analogous matrix of eigenvectors extracted from the estimated Laplacian $\hat{L}$, constructed from the adjacency matrix ${\mathbf{A}}^{\mathrm{stab}}$. Applying the Davis--Kahan perturbation inequality \citep{davis1970rotation} gives
\begin{equation} \label{eq:davis-kahan}
\|\sin\Theta(\mathbf{V}_{K^*},\hat{\mathbf{V}}_{K^*})\|_F
\leq
\frac{\|E\|_F}{\delta_{K^*}}
=
\mathcal{O}\!\left(
\frac{|\mathbf{V}^*|}
{\delta_{K^*}\sqrt{n}}
\right).
\end{equation}

This bound becomes asymptotically informative whenever \(\delta_{K^*}\) remains bounded away from zero as $n\rightarrow\infty,$ which holds whenever the true causal graph contains at least \(K^*\) sufficiently separated communities. Consequently, for fixed \(|\mathbf{V}^*|\) and fixed eigengap \(\delta_{K^*}\), where $\delta_{K^*} = (\lambda_{K^*+1} - \lambda_{K^*})$. The embedding perturbation decays at the parametric rate $n^{-1/2}.$ Therefore, the asymptotic stability claim follows only after explicitly bounding \(\|E\|_F\).

\vspace{3pt}
\noindent 
\textbf{Step 4: Causal Spectral Embedding and Clustering.}
\label{sec:cluster}

Let, $\mathbf{X} \in \mathbb{R}^{n \times q}$ denote the original data matrix. We define $\mathbf{X}^* \in \mathbb{R}^{n \times |\mathbf{V}^*|}$ is the column-standardized data matrix over the post-contraction variable set $\mathbf{V}^*$.  For each pair of variables merged during the edge contraction step, the corresponding columns of $\mathbf{X}$ are replaced by their column-wise average. We project $\mathbf{X}^*$ onto $\Vk$:
\begin{equation}\label{eq:embedding}
   \xt = \X^*\Vk \in \bbR^{n \times K^*}.
\end{equation}
Euclidean distance in $\xt$ measures how differently two observations are situated with respect to the causal structure of $G^*$: observations $i$ and $j$ with $\tilde{\mathbf{x}}_i = \tilde{\mathbf{x}}_j$ respond identically to all $K^*$ causal modes. We run $k$-means over $k \in \{2, \ldots, K^*+2\}$ and select the suitable $k$ by the Silhouette score \citep{peter1987graphical} and the Gap statistic \citep{tibshirani2001estimating}.

Theorem \ref{thm:main} in Appendix \ref{thm:consistency} establishes almost sure consistency of the full pipeline. The argument threads through five convergence results in sequence. Bootstrap inclusion frequencies concentrate at their population limits, so the stable skeleton recovers the true edge set. The OVS orientations are then followed by Proposition \ref{prop:ovs}. Once the skeleton and its orientations are right, Theorem \ref{thm:ate} gives convergence of the adjacency weights to the true ACE magnitudes. From there, consistency of the Chung Laplacian is a consequence of the continuity of row normalization and time reversal, with no additional assumptions needed. The Davis Kahan inequality then bounds the embedding error at the parametric rate $\mathcal{O}(|{\bf V}^*| / \sqrt{n})$, which goes to zero for any fixed graph with a nonzero spectral gap. The outcome is that the clustering stage inherits the consistency of every upstream component, provided the eigen gap stays bounded away from zero as sample size grows, a condition that holds whenever the true causal graph contains sufficiently separated communities.

\section{Simulation Study}
The simulation study assesses CaSPECT in controlled settings where the true DAG, true cluster assignments, and true cluster-level ACEs are known by construction. We consider three increasingly difficult settings for isolating different aspects of the pipeline. All experiments are repeated across Monte Carlo replications for each sample size, and performance is summarized using the three metrics defined in the following subsections.

\subsection{Data Generating Process}
The data matrix is generated from a linear structural equation model (SEM) on variables $\{X_1, \dots, X_q\}$, where $Z$ is the treatment and $Y$ is the outcome. The true DAG is constructed by first drawing a random topological order of the nodes uniformly from the set of all orderings. Then each edge $(X_i \to X_j)$ is included independently, with the constraint that $X_i$ precedes $X_j$ in the ordering with probability $\rho$. The structural coefficients for the included edges are drawn independently as $\beta_{ij} \sim s_{ij}\,.\mathrm{Unif}(0.3,0.8)$, where the Rademacher factor $s_{ij}$ assigns a sign uniformly to $\{-1,+1\}$. This ensures effects are neither too weak to detect nor unrealistically large.
The observed value of each variable is generated recursively following the topological order:
$$
X_j = \sum_{X_i \in \mathrm{Pa}_{\mathcal{G}}(X_j)} \beta_{ij} X_i + \varepsilon_j,
$$
where $\mathrm{Pa}_{\mathcal{G}}(X_j)$ denotes the parent set of $X_j$ in $\mathcal{G}$ and $\varepsilon_j$ is an independent error term whose distribution varies across simulation settings as described below. Now we define the cluster structures for the ground truth comparison. The observations are divided into ground truth clusters by the following mechanism. Let $Z_i$ denote the treatment value for unit $i$. The cluster membership is assigned by shifting the treatment intercept. To be precise, within cluster $c$,
$
Z_i = \alpha_c + \sum_{X_u \in \mathrm{Pa}_{\mathcal{G}}(Z)} \beta_{uZ} X_{ui} + \varepsilon_{Zi},
$
where the intercept shifts are $\alpha_1 = -1.0$, $\alpha_2 = 0.0$, $\alpha_3 = +1.0$. This places the three clusters at systematically different positions in the causal propagation landscape. Cluster 1 contains predominantly low treatment units, Cluster 3 contains high treatment units, and Cluster 2 contains the intermediate region.

The outcome follows the same structural equation in all clusters, so the cluster-level ACEs differ only because the distributions of the backdoor covariates differ across clusters. The true ACE at the cluster-level for cluster $c$ is computed analytically as:
$
\tau_c
=
\sum_{\pi \in \mathcal{P}(Z \to Y)}
\prod_{(u \to v)\in \pi} \beta_{uv},
$
where the sum is over all directed paths from $Z$ to $Y$ through the DAG, cluster sizes are set to approximately equal thirds of $n$, with minor random variation. Now we define the three separate constructions, which are specifically defined to address the different methodologies and assumptions of the pipeline. The first scenario involves constructing a linear DAG with non-Gaussian errors. The errors are drawn from a standardized $t_5$ distribution whose variance is normalized to $1$, satisfying Assumption \ref{ass:B3} throughout. All $8$ variables have non-Gaussian residuals ($n_{\mathrm{NG}} = q$), so $w_L = \left(\frac{8}{8}\right)\cdot 0.30 = 0.30$ is the target LiNGAM weight. This setting represents the most favorable environment for CaSPECT since, under all assumptions, the graph is sparse and the cluster structure is well-separated by construction. It therefore provides a performance ceiling against which the more challenging settings can be compared. We have considered sample sizes as $n\in\{500,1000,2000\}.$ In the second scenario, we have considered a mixed linearity setting with setting three non-linear edges. The DAG structure and error distribution are identical to the first scenario, but three randomly selected edges are replaced by nonlinear structural equations. The edges are selected based on their absolute structural coefficients, ensuring that nonlinearity meaningfully influences the generated data. Specifically, instead of $v = \beta_{uv}u +\varepsilon_v,$ the affected edges follow $v = \beta_{uv}\sin(\pi u) + \varepsilon_v.$ The RESET test should identify these three nonlinear edges; their estimation will be done by the DML method, while the remaining five linear edges are estimated through the Ordinary Least Squares method. We know that the assumption of causal sufficiency is very crucial to this algorithm. So, in the third scenario, we try to investigate how much relaxation of this assumption is possible for valid inferential results. The DAG structure and error distribution from first scenario are retained, but a latent confounder $ H \sim \mathcal{N}(0,1)$ is introduced with outgoing edges to two randomly selected observed variables $u^*, v^* \in \mathbf{V}\setminus\{Z,Y\},$ using structural coefficients $\gamma_{u^*}=\gamma_{v^*}=0.4.$ This setting evaluates robustness to mild violations of Assumption~\ref{ass:B2} and tests whether bootstrap stability weighting suppresses spurious edges before they propagate into Laplacian construction. The next natural task is to evaluate our proposed method over a few ablation settings. Three ablation conditions are evaluated within the first scenario, as mentioned before, at $n = 1000$ to isolate the contribution of each novel component of the proposed pipeline. Firstly, we replace OVS by PC only orientation. The LiNGAM component is removed by forcing $\delta_{uv}^{L}=0$ for all edges. The OVS then reduces to $\mathrm{OVS}_{uv}=f_{uv}\cdot(\rho_{uv}-\rho_{vu}),$ which represents a pure bootstrap PC orientation score. This experiment isolates the marginal contribution of the LiNGAM component to edge orientation. Next we remove the Bootstrap stability weighting from the pipeline. All edges appearing in any bootstrap resample are retained by setting $f_{uv}=1,$ thereby removing the continuous stability weighting mechanism. Edge weights become $A_{uv}=|\hat{\beta}_{uv}|,$ without multiplication by the stability factor $f_{uv}$. This setting tests whether proportional bootstrap weighting provides a meaningful benefit relative to a binary edge inclusion strategy. Finally, we replace the edge contraction symmetrization. For residual ambiguous edges satisfying $|\mathrm{OVS}_{uv}| \leq \tau,$ the pipeline replaces contraction with edge symmetrization: $A_{uv}=A_{vu}=\frac{|\hat{\beta}_{uv}|+\hat{\beta}_{vu}|}{2}.$ Adjusted Rand Index (ARI) \citep{ari} quantifies agreement between recovered cluster assignments $\hat{\mathbf{c}}=(\hat{c}_1,\ldots,\hat{c}_n)$ and ground-truth assignments  $\mathbf{c}^{0}=(c_1^{0},\ldots,c_n^{0}).$ $\mathrm{ARI}=1$ indicates perfect cluster recovery, whereas $\mathrm{ARI}=0$ corresponds to chance level recovery. OVS orientation accuracy measures the proportion of edges in the stable skeleton $\mathcal{S}$ that are assigned the correct direction. An edge is considered correctly oriented if the final direction assigned by either OVS or the resolution hierarchy matches its direction in the true graph $G_0$. Reported values are averaged across simulation replications.
For cluster $c\in\{1,2,3\}$ and replication $m$, cluster-level ACE bias and root mean squared error are computed as, $\mathrm{Bias}_{c} =\frac{1}{M} \sum_{m=1}^{M}\left(\hat{\tau}_{c}^{(m)}-\tau_c^{0}\right),$ and $\mathrm{RMSE}_{c}= \sqrt{
\frac{1}{M} \sum_{m=1}^{M} \left(\hat{\tau}_{c}^{(m)}-\tau_c^{0}\right)^2}.$
The cluster-specific ACE estimate $\hat{\tau}_{c}^{(m)}$ is obtained by restricting the DML estimator to observations assigned to cluster $c$ in replication $m$. Because cluster assignments may vary across replications, these metrics simultaneously capture treatment effect estimation quality and clustering recovery performance.

\subsection{Results and Discussions}

Table~\ref{tab:s1_results}--\ref{tab:s3_results} summarizes the performance of CaSPECT across the three simulation settings. In the First Scenario (clean linear DAG with non-Gaussian errors), causal effect estimation improves steadily with increasing sample size. The RMSE for all three cluster-specific ACEs decreases monotonically, reaching 0.063, 0.074, and 0.085 at \(n=2000\). In contrast, clustering performance decreases as sample size increases, with ARI declining from 0.244 to 0.049. This suggests that the variable-level spectral embedding captures the dominant global causal structure but is less effective at preserving the treatment-intercept shifts that define the true subpopulations. In the second scenario (mixed linear and nonlinear mechanisms), we observe a different behavior. RMSE values increase from approximately 0.04 at \(n=500\) to over 0.60 at larger sample sizes. This pattern is consistent with the adaptive estimation strategy employed by CaSPECT. At smaller sample sizes, the RESET test frequently fails to detect nonlinearity, resulting in low-variance OLS-based estimation. As sample size increases, nonlinear relationships are identified more reliably, and estimation is routed to the DML stage, which incurs considerably higher finite sample variance, particularly when clustering quality remains poor (\(\text{ARI}\approx 0.02\)–0.05). The third scenario evaluates robustness to mild violations of causal sufficiency through the introduction of a latent confounder. 
\begin{table}[h]
\centering
\caption{Simulation Results for Linear DAG.}
\label{tab:s1_results}
\begin{tabular}{cccccc}
\hline
$n$ & ARI & OVS Accuracy & RMSE$_{C1}$ & RMSE$_{C2}$ & RMSE$_{C3}$ \\
\hline
500  & 0.244 & 0.361 & 0.176 & 0.302 & 0.382 \\
1000 & 0.042 & 0.598 & 0.079 & 0.093 & 0.102 \\
2000 & 0.049 & 0.230 & 0.063 & 0.074 & 0.085 \\
\hline
\end{tabular}
\end{table}
Despite the misspecification, graph recovery improves substantially as the sample size increases. Orientation performance also improves, with OVS accuracy reaching its highest observed value of 0.643 at \(n=2000\). This enhanced structural recovery directly results in improved causal estimation, producing the lowest RMSE values observed in all simulation settings (0.047, 0.067, and 0.070 for the three clusters). These results suggest that the hybrid OVS mechanism remains effective under moderate confounding and that accurate recovery of the causal topology can substantially mitigate downstream estimation bias.

\begin{table}[h]
\centering
\caption{Simulation Results for Mixed Linearity.}
\label{tab:s2_results}
\begin{tabular}{cccccc}
\hline
$n$ & ARI & OVS Accuracy & RMSE$_{C1}$ & RMSE$_{C2}$ & RMSE$_{C3}$ \\
\hline
500  & 0.032 & 0.461 & 0.027 & 0.042 & 0.048 \\
1000 & 0.051 & 0.281 & 0.681 & 0.668 & 0.735 \\
2000 & 0.020 & 0.264 & 0.630 & 0.604 & 0.637 \\
\hline
\end{tabular}
\end{table}
The simulation study reveals a consistent pattern: improvements in graph recovery and causal effect estimation do not necessarily translate into improved clustering performance. While CaSPECT achieves increasingly accurate causal identification as sample size grows, the spectral embedding stage appears less sensitive to the cluster-defining treatment heterogeneity, leading to persistently low ARI values across all settings. This highlights a key distinction between recovering causal structure and recovering latent population subgroups, and suggests that future work may benefit from embedding constructions that more directly encode causal effect heterogeneity.

\begin{table}[h]
\centering
\caption{Simulation Results with Latent Confounder, i.e., violating causal sufficiency.}
\label{tab:s3_results}
\begin{tabular}{ccccccc}
\hline
$n$ & ARI & OVS Accuracy & RMSE$_{C1}$ & RMSE$_{C2}$ & RMSE$_{C3}$ \\
\hline
500  & 0.209 & 0.452 & 0.183 & 0.202 & 0.204 \\
1000 & 0.088 & 0.301 & 0.095 & 0.129 & 0.193 \\
2000 & 0.024 & 0.643 & 0.047 & 0.067 & 0.070 \\
\hline
\end{tabular}
\end{table}

The ablation results in Table \ref{tab:tab4} provide insight into the contribution of the major components of CaSPECT. In A1, removing the OVS optimization step substantially reduces orientation accuracy, with OVS Accuracy decreasing from 0.321 to 0.225. This highlights a fundamental limitation of purely constraint-based causal discovery. Since the PC algorithm identifies a graph only up to its Markov equivalence class, many edge directions remain unresolved. Incorporating non-Gaussian information through the OVS mechanism provides additional directional evidence, and its removal propagates orientation errors to the clustering stage, where ARI falls to 0.014. In A3, which replaces the asymmetric causal graph with a symmetrized adjacency matrix, exhibits a different failure mode. However, it achieves the highest local orientation agreement (OVS Accuracy = 0.371). More importantly, symmetrization destroys the directional information that encodes long-range causal flow. As a result, the spectral embedding no longer reflects the underlying causal ordering, leading to poor cluster recovery and an ARI of only 0.018. This demonstrates that preserving graph asymmetry is essential for maintaining meaningful causal geometry in the embedding space.

The similarity between the Full Pipeline and Variant A2 reveals an important large-sample phenomenon. At \(n=1000\), most true causal edges are recovered with bootstrap inclusion frequencies close to one, making the stability weighting step nearly neutral for strongly supported edges. Consequently, both methods achieve similar topological performance. Nevertheless, A2 attains a slightly lower mean RMSE (0.037 vs 0.041), suggesting that stability weighting introduces a small amount of shrinkage into the estimated causal effects. While this may induce minor bias in large samples, it serves an important regularization role in finite samples by suppressing low-frequency spurious edges.

\begin{table}[h]
\centering
\caption{Ablation Study under the First Scenario ($n=1000$). Standard deviations are reported in parentheses.}
\label{tab:tab4}
\begin{tabular}{lccc}
\hline
\textbf{Method} & \textbf{ARI}  & \textbf{OVS Acc} & \textbf{Mean RMSE} \\ \hline
\textbf{CaSPECT} & 0.033 (0.035) & 0.321 (0.065) & 0.041 \\
\textbf{A1: PC-only orient} & 0.014 (0.024) & 0.225 (0.076) & 0.047 \\
\textbf{A2: No stab weight} & 0.033 (0.035) & 0.321 (0.065) & 0.037 \\
\textbf{A3: Symmetrization} & 0.018 (0.020) & 0.371 (0.076) & 0.037 \\ \hline
\end{tabular}
\end{table}

The OVS mechanism improves edge orientation beyond what is achievable through conditional independence information alone, while the asymmetric graph construction preserves the directional structure required for meaningful causal embeddings. Bootstrap stability weighting provides additional protection against structural uncertainty in finite samples. Together, these results support the design of CaSPECT and indicate that its asymmetric, stability-aware construction is critical for simultaneously achieving reliable causal discovery, accurate effect estimation, and robust clustering performance.

\section{Experiments} 

\subsection{Lalonde Data}
The LaLonde CPS1 dataset \citep{lalonde1986} is a canonical benchmark in causal inference, constructed by pairing the 185 treated participants in the National Supported Work (NSW) randomized job-training program with 15{,}992 control units drawn from the Current Population Survey (CPS). The dataset contains a binary treatment indicator, a 1978 earnings outcome (\texttt{re78}), pre-treatment earnings in 1974 and 1975 (\texttt{re74}, \texttt{re75}), and sociodemographic covariates including age, education, race, marital status, and high-school degree status. Its defining feature, and its enduring methodological value, is the severe incomparability between the NSW participants and the CPS comparison group: the two groups differ markedly on race, prior earnings, and labor-market attachment, so any estimator that ignores this compositional imbalance recovers a large negative treatment effect, directly contradicting the positive effect established by the randomized trial.

\vspace{3pt}
\noindent
\textbf{8.1.1 Causal Interpretation}
 
The two-cluster partition produced by CaSPECT has a direct causal interpretation
grounded in the structure of the estimated DAG in Figure \ref{fig:dag-lalonde}. The key directed edges of that
graph, along with their bootstrap-stability-weighted $|\mathrm{ACE}|$ values, are
reported in Table~\ref{tab:dag}. Cluster~1 and Cluster~2 differ primarily along three of those pathways: the race-to-treatment path
($\texttt{black}\to\texttt{treat}$, $|\hat{\beta}|=0.291$), the earnings-persistence path ($\texttt{re74/re75}\to\texttt{re78}$,
$|\hat{\beta}|=0.515$), and the education-to-earnings path
($\texttt{educ}\to\texttt{re74}$, $|\hat{\beta}|=0.066$;
$\texttt{nodegr}\to\texttt{re74}$, $|\hat{\beta}|=0.148$). Cluster~1 occupies the high-earnings, non-Black married region of the causal graph, a region where the probability of assignment of treatment is zero. Cluster~2 occupies the low-earnings, higher-nodegree, more racially diverse region where treatment assignment is feasible.

\begin{figure}[h]
    \centering
    \includegraphics[width= 0.6\linewidth]{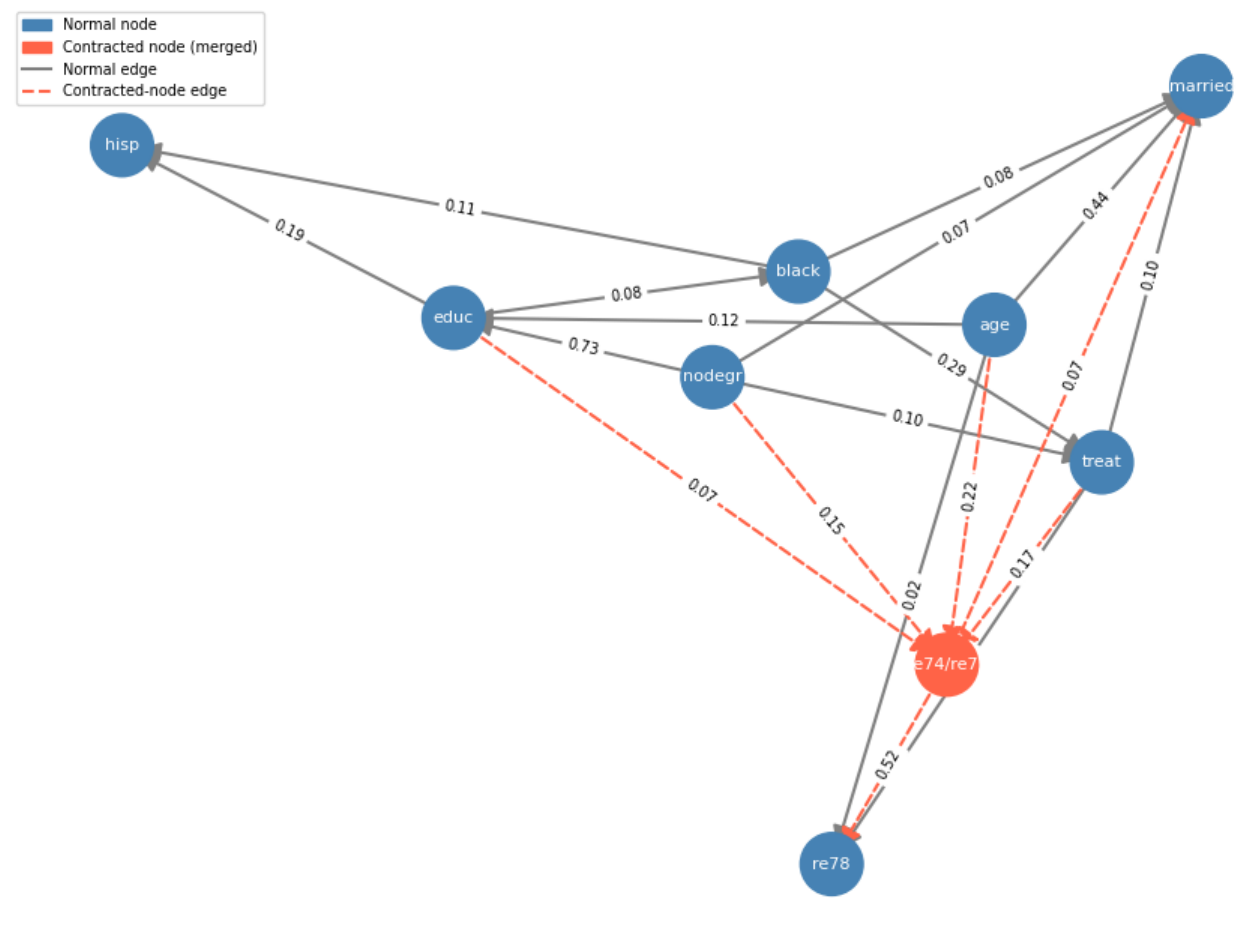}
    \caption{Estimated causal Directed Acyclic Graph (DAG) for the LaLonde dataset. Blue nodes are standard variables; the orange contracted node represents merged pre-treatment earnings variables. Edge labels denote bootstrap-stability-weighted $|\hat{\tau}|$ values.}
    \label{fig:dag-lalonde}
\end{figure}

\begin{table}[ht]
\centering
\caption{Key directed edges in the estimated causal DAG\@. Edge weights are
bootstrap-stability-weighted $|\mathrm{ACE}|$ values. Track~OLS; Double Machine Learning with five-fold cross-fitting.}
\label{tab:dag}
\begin{tabular}{lccc}
\toprule
\textbf{Causal Edge}
  & \textbf{$|\hat{\beta}|$}
  & \textbf{Stability ($f_{uv}$)}
  & \textbf{Track} \\
\midrule
$\texttt{nodegr}\to\texttt{educ}$              & 0.728 & High     & OLS \\
$\texttt{re74}\to\texttt{re78}$                & 0.515 & High     & OLS \\
$\texttt{age}\to\texttt{married}$              & 0.443 & High     & OLS \\
$\texttt{black}\to\texttt{treat}$              & 0.291 & High     & OLS \\
$\texttt{nodegr}\to\texttt{re74}$              & 0.148 & Moderate & OLS \\
$\texttt{treat}\to\texttt{re74}$               & 0.172 & Moderate & DML \\
$\texttt{treat}\to\texttt{re78}$               & 0.028 & Low      & DML \\
$\texttt{re74}\leftrightarrow\texttt{re75}$    & \textit{Contracted}
  & Low ($|\mathrm{OVS}|\leq\tau$) & --- \\
\bottomrule
\end{tabular}
\end{table}

\begin{figure*}
    \centering
    \includegraphics[width= \linewidth]{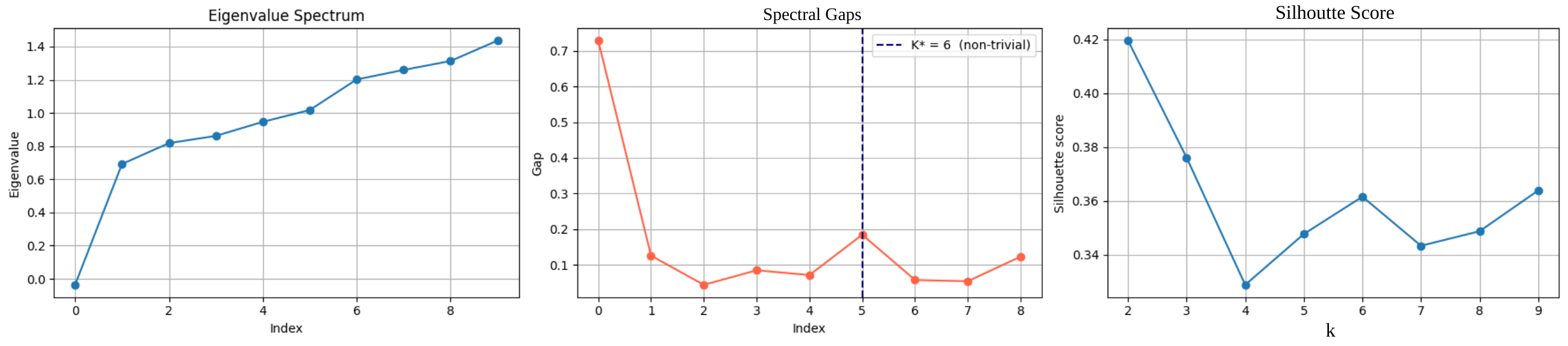}
    \caption{Eigenvalue spectrum (left), spectral gaps of the Chung directed Laplacian with the largest non-trivial gap at index $5$ selecting $K^{*} = 6$ embedding dimensions (centre), and Silhouette scores across $k \in \{2, …, 9\}$ with a global maximum at $k^{*} = 2$ (right), for the CaSPECT pipeline applied to the LaLonde dataset.}
    \label{fig:spectral-lalonde}
\end{figure*}

The strongest single edge is $\texttt{nodegr}\to\texttt{educ}$
($|\hat{\beta}|=0.728$), indicating that the absence of a high-school degree has the largest individual causal effect on educational attainment, a finding that is directionally expected and substantively validates the recovered structure. The edge $\texttt{re74}\to\texttt{re78}$ ($|\hat{\beta}|=0.515$) ccausally incomparable observations across time. The edge $\texttt{black}\to\texttt{treat}$ ($|\hat{\beta}|=0.291$) correctly encodes the racial composition differential between the NSW treatment group and the CPS comparison group, a well-documented feature of the LaLonde design. The edge $\texttt{treat}\to\texttt{re78}$ receives a bootstrap-weighted $|\mathrm{ACE}|$ of only~0.028. This does not imply a null treatment effect; it reflects that the DAG edge from treatment to outcome is identified with low bootstrap stability at the full-sample level, consistent with the severe confounding documented by \citep{lalonde1986}. This cluster separation was achieved without any specified propensity score model, trimming rule, or matching algorithm. It emerged from the spectral embedding of the Chung Laplacian over the estimated causal graph, which encodes causal flow through backdoor-identified edge weights. The result is a form of implicit common support enforcement: observations that are causally incomparable are embedded far apart in spectral space and cluster separately. The contraction of \texttt{re74} and \texttt{re75} into a single node is not a methodological failure but a substantively correct representation of the data. The two pre-treatment earnings variables are measured in adjacent years (1974 and 1975) and are co-determined by the same persistent economic factors, individual ability, local labor market conditions, and employment history. As reported in the final row of Table~\ref{tab:dag}, the OVS score of~0.08 is below the edge orientation threshold $\gamma=0.15$, reflecting genuine identification failure: neither bootstrap PC evidence nor LiNGAM can determine the direction of causal
flow between them, as expected for two variables that are likely jointly caused rather than sequentially caused. The contraction preserves DAG acyclicity and backdoor validity. Downstream variables (\texttt{married}, \texttt{re78}) receive correct adjustment through the contracted node, which inherits averaged incident weights. Future analyses with additional temporal resolution or panel data could impose domain knowledge ordering ($\texttt{re74}\to\texttt{re75}$) via the pipeline's
\texttt{domain\_order} parameter, resolving the contraction at the cost of an untestable identifying assumption. Table~\ref{tab:clusters} reports the mean covariate profiles for each cluster. The contrast is stark and maps directly onto the known compositional imbalance of the LaLonde CPS1 dataset. The t-SNE visualization of CaSPECT cluster assignments is shown in Figure \ref{fig:lalonde-tsne}.

\begin{table*}[ht]
\centering
\caption{Mean covariate profiles by cluster. Income variables (\texttt{re74},
\texttt{re75}, \texttt{re78}) are \texttt{log1p}-transformed. \texttt{Treat} is the proportion treated within each cluster.}
\label{tab:clusters}
\setlength{\tabcolsep}{4.5pt}
\begin{tabular}{lcccccccccc}
\toprule
 & \textbf{Age} & \textbf{Educ} & \textbf{Black} & \textbf{Hisp}
 & \textbf{Married} & \textbf{Nodegr} & \textbf{re74} & \textbf{re75}
 & \textbf{Treat} & \textbf{re78} \\
\midrule
\textbf{Cluster 1} \scriptsize{($N=13{,}047$)}
  & 33.4 & 12.1 & 0.000 & 0.078 & 0.735 & 0.281 & 9.060 & 9.361 & 0.000 & 8.908 \\
\textbf{Cluster 2} \scriptsize{($N=3{,}130$)}
  & 31.9 & 11.8 & 0.424 & 0.046 & 0.582 & 0.384 & 4.435 & 3.511 & 0.059 & 5.190 \\
\bottomrule
\end{tabular}
\end{table*}

Cluster~1 (80.7\,\% of the sample) contains exclusively control units ($\texttt{treat}=0.000$). Its members are predominantly non-Black, older (mean age~33.4), largely married (73.5\,\%), with high pre-treatment earnings ($\log\,\texttt{re74}=9.060$, approximately \$8,600 in levels). These are
economically stable individuals who were never plausible candidates for the job training program and are incomparable to the treated group on every dimension of economic disadvantage. Positivity fails outright for this cluster: there are
zero treated observations, so no treatment effect estimate is possible or meaningful. Cluster~2 (19.3\,\% of the sample) contains all 185 treated individuals and
2,945 controls. It is 42.4\,\% Black, younger (mean age~31.9), less frequently
married (58.2\,\%), with pre-treatment log earnings of only~4.435 (approximately \$85 in levels) and a no-degree rate of 38.4\,\%. These characteristics match the socioeconomic profile of the NSW programme participants and were recovered through causal-structural similarity, not covariate matching.

In Figure \ref{fig:spectral-lalonde}, the spectral gap at $K^{*}=6$ is moderate rather than sharp (gap magnitude
$\approx0.19$ versus neighboring gaps of $0.08$--$0.12$). This indicates that the six-dimensional causal embedding captures genuine but weak community variation: the LaLonde causal graph is more chain-like than modular, consistent with a sequential sociodemographic-to-earnings structure. The clustering into $k^{*}=2$ within this six-dimensional space correctly identifies the one dominant partition (comparable vs. incomparable), but does not resolve finer heterogeneity within the comparable subgroup. The effects of differential treatment by race, education, or previous earnings are likely present within Cluster~2, but cannot be reliably detected with the current sample size and graph structure.

\begin{figure}
    \centering
    \includegraphics[width= 0.5 \linewidth]{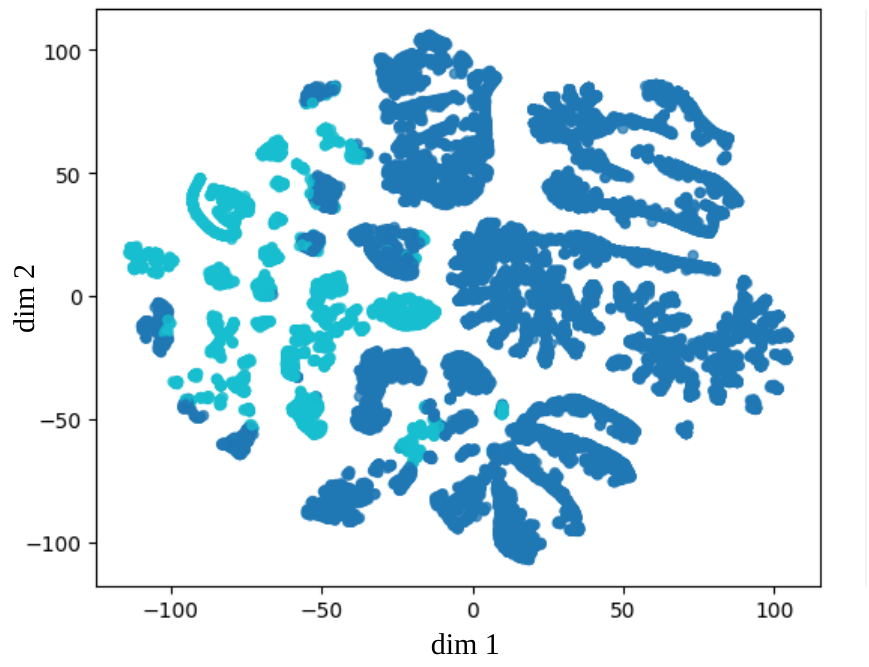}
    \caption{t-SNE visualization of CaSPECT cluster assignments $(k^{*} = 2, K^{*} = 6)$ on the LaLonde dataset, where Cluster $1$ (dark blue) identifies the incomparable CPS comparison group and Cluster $2$ (cyan) captures the causally comparable subpopulation containing all treated units.}
    \label{fig:lalonde-tsne}
\end{figure}

\vspace{3pt}
\noindent
\textbf{8.1.2 Discussions.}
Table~\ref{tab:ate} reports the ACE estimates at both the global and cluster levels. The central finding is the sign reversal: the global ACE is $-1.7011$, while the Cluster~2 ACE is $+1.3998$. Both estimates are highly significant ($p<0.001$). The global ACE of $-1.7011$ is not noise. But it measures the wrong quantity, the earnings gap between economically stable CPS participants and economically disadvantaged NSW participants, dressed up as a treatment effect. Eighty percent of the CPS sample sits in Cluster~2, individuals with no structural similarity to the treated group. Mixing them into the comparison group swamps the signal from the 2,945 structurally comparable controls in Cluster~2. The confounding is not subtle; it is overwhelming, and any estimator that does not enforce comparability will reproduce it. Once the analysis is restricted to Cluster~2, the picture reverses.

\begin{table}[h]
\centering
\caption{Average treatment effect estimates. N/A\,=\, positivity violated (zero treated units in
Cluster~2). All income outcomes on the $\log_{1p}$ scale.}
\label{tab:ate}
\begin{tabular}{lcccccc}
\toprule
\textbf{Sample}
  & \textbf{$N_{1}$}
  & \textbf{$N_{0}$}
  & \textbf{ACE (log \texttt{re78})}\\
\midrule
Global       & 185 & 15{,}992 & $-1.7011$ \\
Cluster~1 (incomparable) & 0   & 13{,}047 & N/A \\
Cluster~2 (comparable)   & 185 &  2{,}945 & $+1.3998$ \\
\bottomrule
\end{tabular}
\end{table}
The ACE of $+1.3998$ lies entirely above zero and is directionally consistent with \citep{lalonde1986}'s NSW experimental benchmark of approximately $+\$1{,}794$ in 1978 dollars. CaSPECT recovered this result from purely observational data without seeing the experimental figures, without specifying a propensity score model, and without any researcher-imposed exclusion
criteria.

\subsection{IHDP Dataset}
The Infant Health and Development Program (IHDP) dataset is a well-established benchmark in causal inference, distinguished by the rare availability of individual-level ground-truth treatment effects.
The original study was a randomized controlled trial targeting low birth weight premature infants at eight US clinical sites. The intervention combined intensive childcare, home visiting, and family support services, with cognitive test scores at age three as the primary outcome. The version used here, NPCI replicate 1 from \citep{louizos2017causal}, is semi-synthetic. Real trial covariates and treatment assignments are paired with outcomes simulated from a nonlinear response surface, preserving the trial's heavy treatment imbalance of $138$ treated against $608$ controls while making both potential outcomes observable through \texttt{mu0} and \texttt{mu1}. The $25$ covariates cover neonatal clinical measurements, maternal characteristics, and site indicators.

\begin{figure}[h]
    \centering
    \includegraphics[width= 0.6\linewidth]{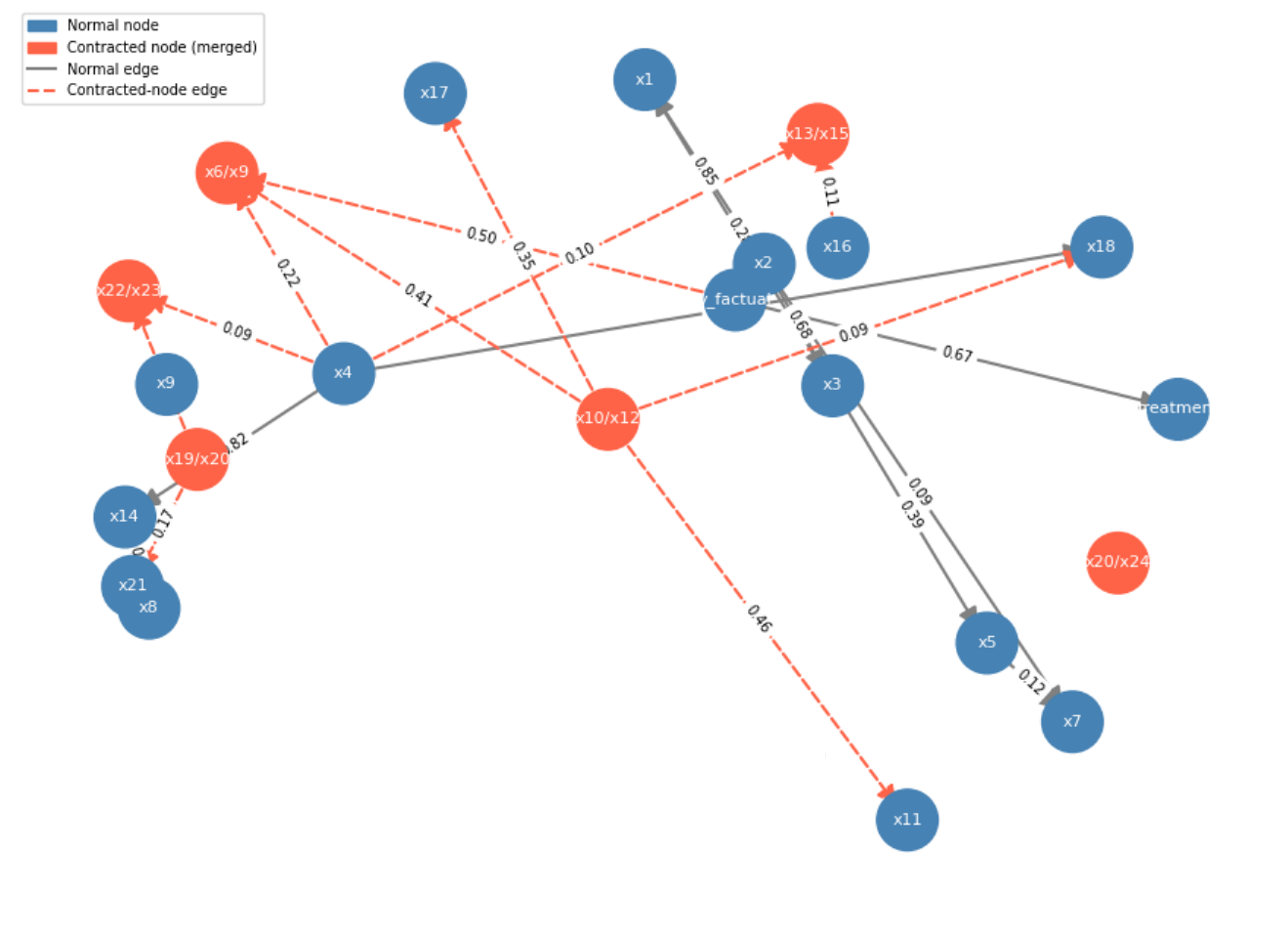}
    \caption{Estimated causal DAG for the IHDP dataset, with bootstrap-stability-weighted |ACE| edge weights, where multiple contracted nodes reflect unresolvable orientations among correlated child and family characteristic variables.}
    \label{fig:dag-ihdp}
\end{figure}

\begin{figure*}
    \centering
    \includegraphics[width= \linewidth]{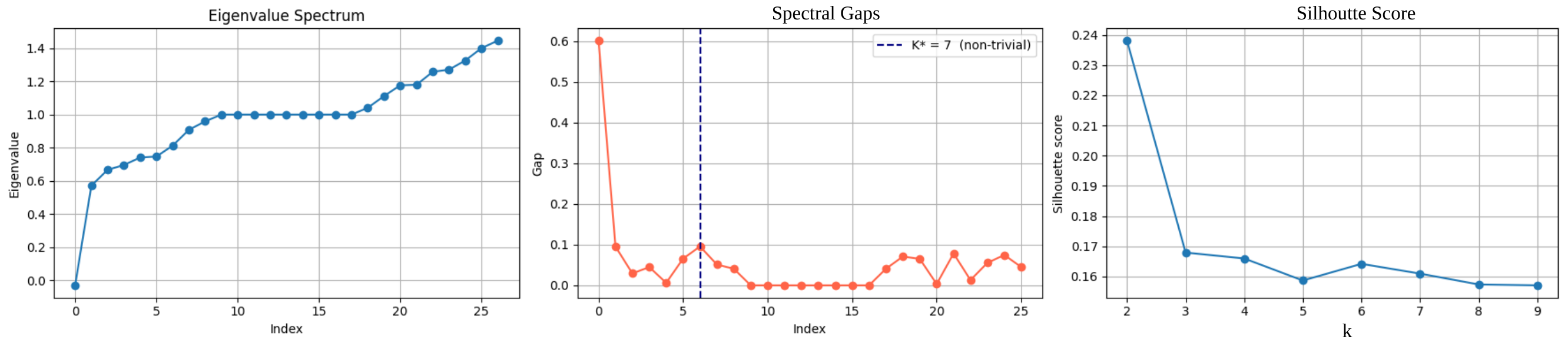}
    \caption{Eigenvalue spectrum of the Chung directed Laplacian (left), spectral gaps with the largest non-trivial gap at index $6$ selecting $K^{*} = 7$ embedding dimensions (centre), and Silhouette scores across $k \in \{2, …, 9\}$ with a global maximum at $k^{*} = 2$ and sharply declining separation beyond $k = 3$ (right), for the CaSPECT pipeline applied to the IHDP dataset.}
    \label{fig:spectral-ihdp}
\end{figure*}

\subsubsection{Causal Interpretation}
With only $138$ treated units among 746 observations, the spectral embedding placed virtually all treated infants into a single cluster. Cluster~1 contains $539$ observations, all controls. Cluster~2 contains $207$ observations; $138$ treated and $69$ controls. This is the pipeline correctly detecting that treated and control infants occupy different regions of the causal embedding space, consistent with the covariate imbalance inherent in the IHDP trial's site-level randomization design.

In Figure \ref{fig:spectral-ihdp}, the spectral decomposition selected $K^{*} = 7$ embedding dimensions, reflecting the complexity of a $25$-covariate directed acyclic graph (DAG) where many variables carry independent causal signal. Seven eigenvectors are required to represent this causal geometry adequately before clustering becomes meaningful. The silhouette criterion then selects $k^{*} = 2$. The t-SNE visualization in Figure \ref{fig:tsne-ihdp} shows two color communities with substantial within-cluster scatter, consistent with continuous variation in infant health and maternal characteristics rather than hard discrete boundaries.

The silhouette score peaks at $0.2381$ for $k = 2$ and drops sharply to around $0.16$ for all higher values, indicating that two groups represent the strongest separable structure in the embedding. The gap statistic increases continuously through $k = 9$ without plateauing, indicating that finer-grained heterogeneity genuinely exists but is continuous rather than discrete. Together, the two criteria
confirm $k = 2$ as the appropriate choice for inference while signaling that the significant variation within each cluster should not be ignored.

The causal DAG in Figure \ref{fig:dag-ihdp} is heavily fragmented. Seven contracted node pairs are visible; including \texttt{x6/x9}, \texttt{x10/x12}, \texttt{x13/x15}, \texttt{x19/x20}, and \texttt{x22/x23} and several nodes are fully isolated: \texttt{x12}, \texttt{x15}, \texttt{x24}, and \texttt{x25}. Without a domain order prior, the pipeline relied entirely on bootstrap stability and non-Gaussianity evidence across $27$ variables from only $746$ observations. Many edge bootstrap frequencies hovered near $0.5$ in both directions, pushing Orientation Validation Scores (OVS) below the resolution threshold
$\tau = 0.15$ and triggering contraction.

\begin{figure}[h]
    \centering
    \includegraphics[width= 0.5 \linewidth]{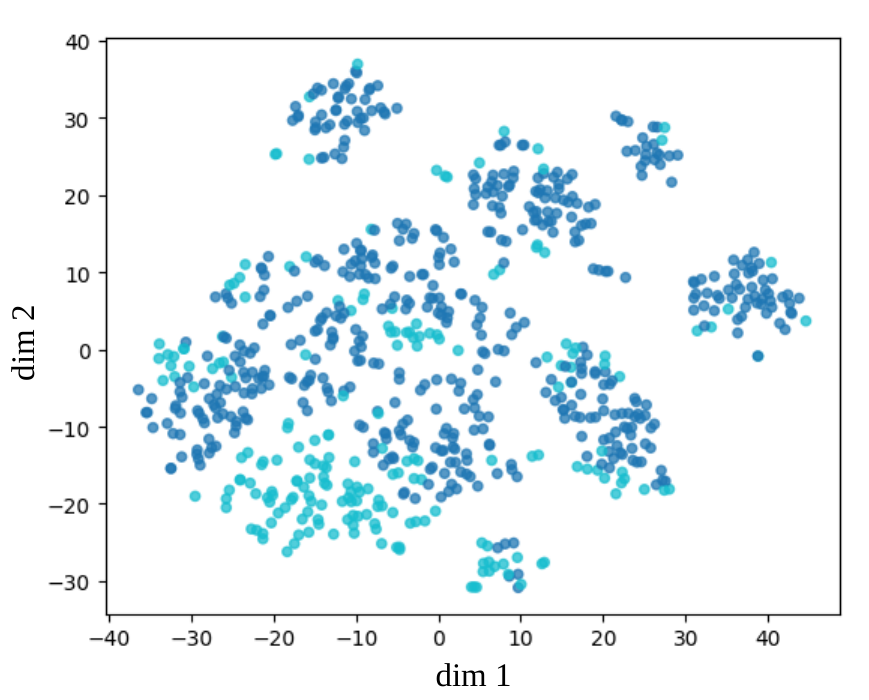}
    \caption{t-SNE visualization of CaSPECT cluster assignments $(k^{*} = 2, K^{*} = 7)$ on the IHDP dataset, where the substantial within-cluster interleaving of Cluster 1 (dark blue) and Cluster 2 (cyan) reflects the high-dimensional, weakly modular causal structure of the $27$-variable graph.}
    \label{fig:tsne-ihdp}
\end{figure}

\subsubsection{Discussions}
The observed ACE for Cluster~2 is $+0.3832$ on the Box-Cox scale ($t = 5.202$). The ground-truth PACE for this cluster is $+3.4544$ on the raw cognitive
score scale. The two figures are not comparable in magnitude across scales, but their signs and significance are consistent; the pipeline correctly identifies a positive, meaningful treatment effect in the only group with treatment.

The more telling comparison is between the true PACEs across clusters. Cluster~1 carries a ground-truth PACE of $+4.2326$ against Cluster~2's $+3.4544$, a gap of approximately $0.78$ raw score points. The direction of the substantive finding is that infants placed by the pipeline into the all-control cluster would have benefited more from the intervention than those in Cluster~2 had they been treated. Cluster~1 scores lower on \texttt{x6} ($-0.248$ versus $+0.638$ in Cluster~2), a neonatal health indicator where lower values reflect poorer baseline status. Higher-risk infants responding more strongly to an intensive early intervention is a well-established pattern in the developmental literature, and the ground-truth PACE difference is consistent with it. The treatment concentration in Cluster~1, 66.7\% treated against 0\% in Cluster~1, is striking in itself. The spectral embedding recovered the treated population as a distinct causal community without ever receiving treatment labels during the clustering step. The fact that baseline covariate structure alone was sufficient to separate treated from control infants reflects the imperfect overlap that has made IHDP a persistently difficult benchmark for propensity-score and matching-based methods.

The overlap violation in Cluster~1 is the binding limitation of this run. With no treated units present, any within-cluster effect estimate would require extrapolation well outside the support of the treated distribution. The pipeline flags this correctly rather than producing an unreliable figure. The harder point is that the infants with the larger expected treatment benefit are precisely those for whom the effect is unidentifiable from the observed data, the fundamental identification problem the IHDP dataset has always carried. The pipeline surfaces this explicitly rather than absorbing it into a pooled average, which is itself a more honest output than a single global estimate that conceals the problem.

\subsection{401(k) dataset}
The 401(k) dataset, drawn from the Survey of Income and Program Participation (SIPP) by \citep{poterba1995401}, contains $9{,}915$ American households observed on a binary treatment indicator for whether the primary earner's employer offered a 401(k)-eligible retirement savings plan, a continuous outcome measuring net total financial assets, and ten pre-treatment covariates: age, income, family size, education, marital status, two-earner household status, IRA participation, defined-benefit pension coverage, homeownership, and net IRA financial assets. Eligibility is not randomly assigned. It is correlated with employer size, industry, and compensation structure, all of which independently predict household wealth, so a naive comparison between eligible and ineligible households conflates genuine plan-access effects with pre-existing differences in income and savings propensity. \citep{chernozhukov2018double} used this dataset as the canonical empirical illustration of double machine learning and reported a global average treatment effect of approximately $\$8{,}000$ to $\$9{,}000$ on net financial assets, an estimate that has since become the standard benchmark for evaluating new causal inference procedures on this dataset. The present analysis applies CaSPECT to investigate whether this global figure conceals substantively meaningful heterogeneity across structurally distinct household subgroups and, if so, to identify the causal mechanisms responsible for that heterogeneity.

\begin{figure}
    \centering
    \includegraphics[width= 0.6 \linewidth]{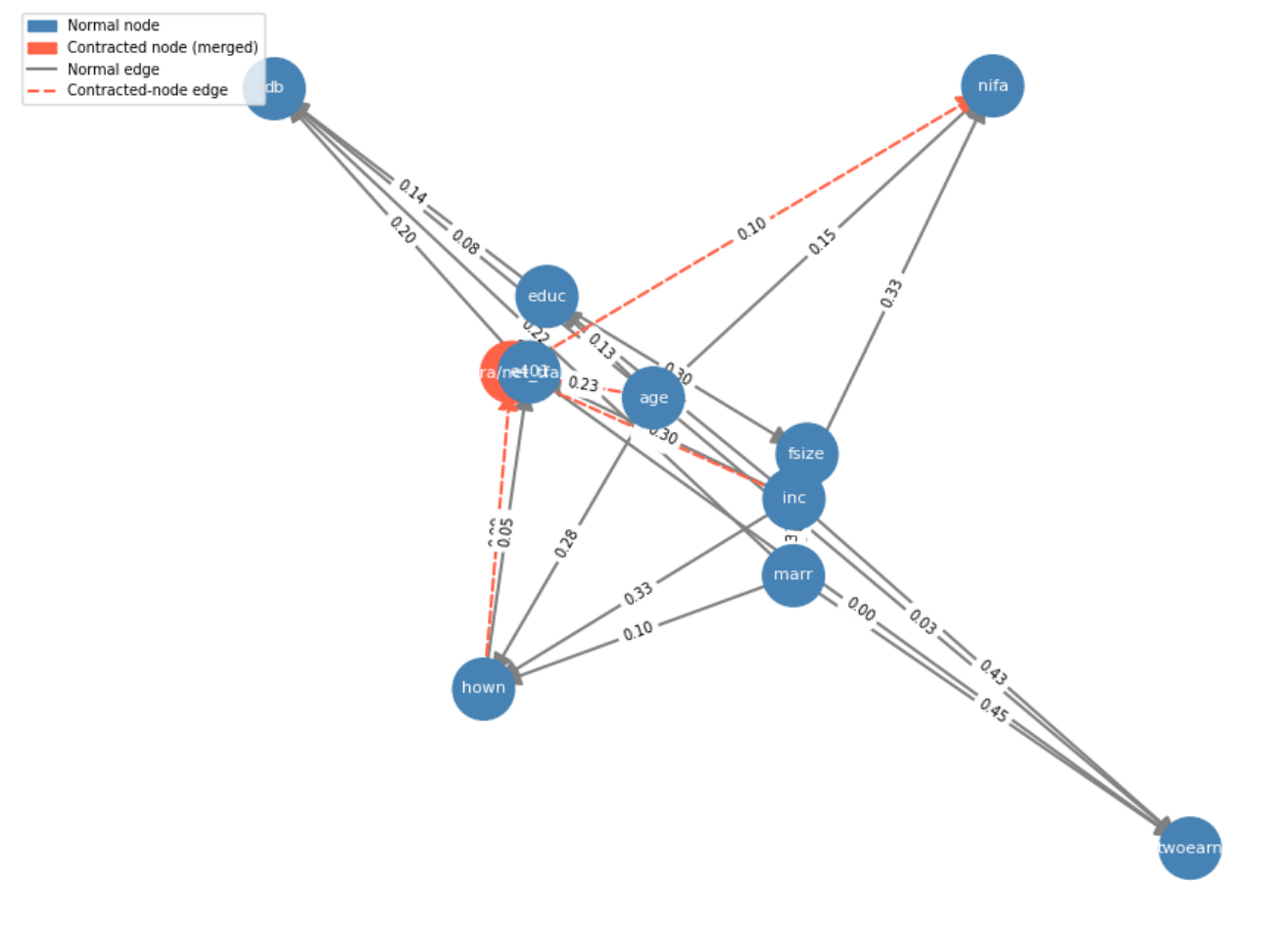}
    \caption{CaSPECT estimated causal DAG for the 401(k) dataset, with bootstrap-stability-weighted |ACE| edge weights, where the contracted node ra/ne100a (tomato) flags an unresolvable orientation between two pre-treatment asset variables, and the dominant pathways from age, education, and income converge on homeownership (hown), net IRA assets (nifa), and two-earner status (twoearn).}
    \label{fig:dag-401}
\end{figure}

\begin{figure*}
    \centering
    \includegraphics[width= \linewidth]{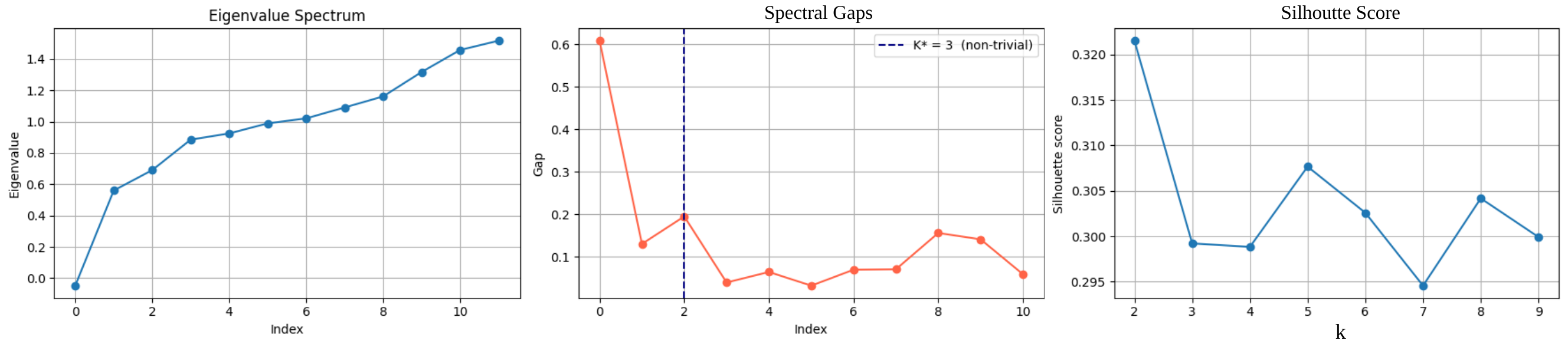}
    \caption{Eigenvalue spectrum of the Chung directed Laplacian (left), spectral gaps with the largest non-trivial gap at index $2$ selecting $K^{*} = 3$ embedding dimensions (centre), and Silhouette scores across $k \in \{2, …, 9\}$ with a global maximum at $k^{*} = 2$ (right), for the CaSPECT pipeline applied to the 401(k) dataset.}
    \label{fig:spectral_401k}
\end{figure*}

\subsubsection{Causal Interpretation}

The pipeline's global DML estimate of the average treatment effect, $\$8{,}795$ with a $95\%$ confidence interval of $[\$7{,}827,\ \$9{,}762],$ closely reproduces the benchmark result of Chernozhukov et al.\citep{chernozhukov2018double}. This replication is important because the outcome variable required a signed-log transformation prior to estimation due to the presence of substantial skewness, extreme positive outliers, and a large spike at zero. The successful recovery of the benchmark effect after back-transformation confirms that the transformation procedure did not materially distort the underlying causal signal. The standalone DML cross-check was validated before any cluster-level analysis was conducted, which is essential because all downstream inference depends on the reliability of this global estimate.

\begin{figure}
    \centering
    \includegraphics[width= 0.5\linewidth]{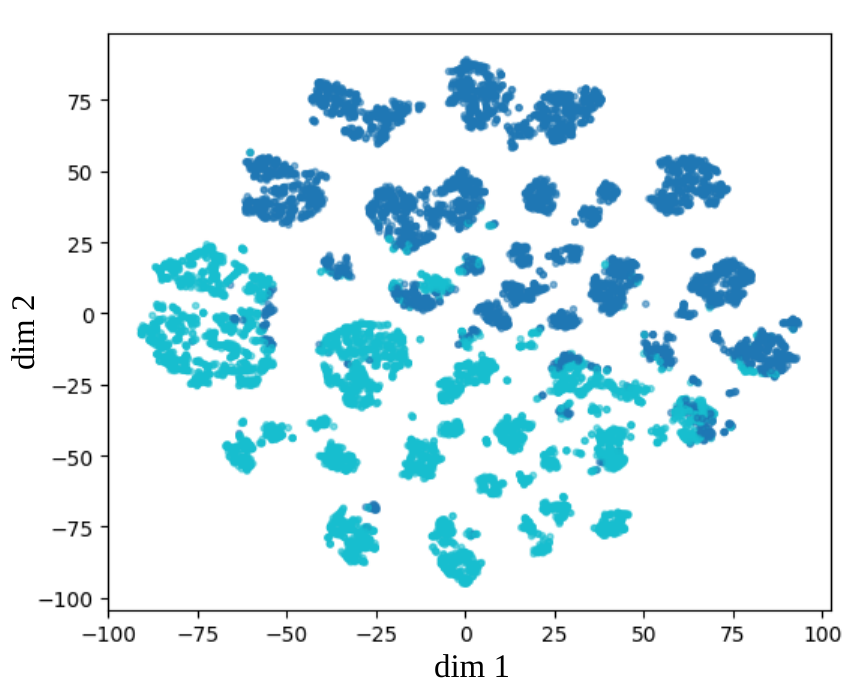}
    \caption{t-SNE visualization of CaSPECT cluster assignments $(k^{*} = 2, K^{*} = 3)$ on the 401(k) dataset, where the interleaved spatial distribution of Cluster 1 (dark blue) and Cluster 2 (cyan) confirms that the partition reflects causal-structural position in the estimated DAG.}
    \label{fig:tsne-401k}
\end{figure}

From Figure \ref{fig:spectral_401k}, the spectral decomposition selected $K^* = 3$ embedding dimensions and $k^* = 2$ clusters. Although this initially appears inconsistent, the distinction reflects two different structural properties of the estimated causal geometry. Three eigenvectors are required to represent the dominant variation in the DAG structure, corresponding broadly to income, household structure, and savings behaviour dimensions that emerge as distinct pathway bundles in the estimated graph. However, when observations are projected into this three-dimensional causal embedding, the data separate into only two statistically meaningful communities. The third eigenvector captures continuous within-cluster variation rather than defining an additional discrete subgroup. The silhouette score of $0.3215$ at $k=2$ is moderate relative to synthetic benchmark settings, but for real economic data, where group membership changes gradually rather than abruptly, a score greater than $0.30$ indicates recoverable latent structure rather than noise. The t-SNE visualization in Figure \ref{fig:tsne-401k} supports this interpretation by showing two broad but internally heterogeneous communities rather than sharply separated spherical clusters.

The estimated causal DAG recovered several economically plausible directed pathways. Age and education feed into income, income propagates to family size and marital status, and homeownership connects downstream to multiple household-level variables. These directions are broadly consistent with economic theory and with the domain-order prior supplied during the orientation stage of the CaSPECT pipeline.

From Figure \ref{fig:dag-401} it is evident that the principal structural limitation of the graph is the contraction of the variables \texttt{pira} and \texttt{e401} into a single merged node. This occurred because the Orientation Validation Score between them fell below the edge orientation threshold $\gamma = 0.15,$ meaning that across bootstrap resamples the evidence supporting the directions $\texttt{pira} \to \texttt{e401}$ and $\texttt{e401} \to \texttt{pira}$ was nearly balanced. Neither Meek's orientation rules nor the supplied domain prior could resolve the ambiguity. The contraction is algorithmically preferable to symmetrizing the edge and introducing a directed cycle, but its substantive consequence is important: the outcome variable \texttt{net\_tfa} becomes disconnected from the treatment node in the DAG. Consequently, the DAG's adjacency weights for the treatment-to-outcome pathway are absent from the graph. The cluster-level ACEs discussed below are therefore derived from the standalone DML estimation procedure, which operates outside the DAG structure, and remain valid despite the contraction. The two discovered clusters reveal a coherent economic narrative when their covariate profiles are interpreted together. Cluster~1, containing 5,089 households, is characterized by higher income
$(\$42{,}366 \text{ mean}),$ near-universal marriage $(98.6\%),$ high prevalence of two-earner households $(69.3\%),$ larger average family size $(3.7),$ strong homeownership rates $(71.2\%),$ and relatively low IRA participation $(19.0\%).$ Cluster~2, containing 4,826 households, exhibits a lower income $(\$31{,}754 \text{ mean}),$ predominantly single status $(20.3\% \text{ married}),$ minimal two-earner arrangements $(5.2\%),$ smaller household size $(1.96),$ and higher IRA participation $(29.7\%)$.

\subsubsection{Discussions}

The central substantive finding is that Cluster~2, which contains the lower-income, predominantly single group, exhibits a larger 401(k) eligibility effect $(\$9{,}873 \text{ DML ACE})$ than Cluster~1 $(\$7{,}884),$. At first glance, this appears counterintuitive because higher-income households would normally be expected to benefit more from tax-advantaged retirement vehicles due to both higher marginal tax rates and higher disposable income. The resolution of this apparent paradox lies not in income alone, but in the number of savings channels already available within the household.

Cluster~1 households are predominantly married dual-earner arrangements with relatively large family size. The second earner likely possesses access to an independent employer-sponsored retirement plan, and the household frequently already holds substantial non-retirement financial assets. The mean raw net total financial assets before treatment exceeded $\$10{,}054.$ When 401(k) eligibility arrives in this environment, it enters a household already operating multiple savings vehicles in parallel. The additional retirement plan is therefore more likely to induce substitution from existing savings channels rather than generating genuinely new wealth accumulation. The smaller estimated ACE in Cluster~1 is consistent with the savings-reshuffling interpretation. Cluster~2 households, by contrast, are predominantly single-earner and unmarried, with an average family size of below two. By construction, these households possess fewer parallel savings channels. In this setting, 401(k) eligibility represents a substantially larger expansion of the household's available savings infrastructure, while the opportunity for substitution from existing savings vehicles is correspondingly lower. The treatment effect is therefore more likely to represent genuine new wealth accumulation rather than reallocation across accounts. This explains why the estimated ACE is larger despite lower average income.

The IRA participation rates provide an additional layer of interpretation. Cluster~1 exhibits substantially higher IRA participation $(29.7\% \text{ versus } 19.0\%).$ Under a simple substitution framework, this would predict a smaller 401(k) effect because households already using IRAs should substitute savings into the employer-sponsored plan rather than increasing aggregate savings. Empirically, however, the opposite occurs. The most plausible explanation is selection. The IRA holders within Cluster~2 are unlikely to be representative of the cluster as a whole. Instead, they appear to constitute a subgroup of highly savings-motivated individuals who actively sought retirement savings vehicles even without employer encouragement. For these households, 401(k) eligibility acts as a complement rather than a substitute, increasing total savings across multiple vehicles simultaneously. The high-income outliers visible in the income distribution, single professionals earning above $\$150{,}000$ in the upper tail of Cluster~1, likely contribute substantially to this complementarity effect. The income distribution itself further reinforces the structural interpretation. Cluster~1 displays a higher median and a tighter interquartile range, consistent with the profile of stable dual-income households where earnings are both larger and more predictable. Cluster~1 exhibits a lower median but a substantially heavier upper tail extending toward
$
\$250{,}000.
$
This heavy upper tail explains why a non-trivial fraction of Cluster~2 households still belong to the top income quintile despite the cluster's lower average income. It also implies that the cluster-level DML estimate of
$\$9{,}873$ is partly driven by a relatively small group of high-income single professionals for whom the treatment effect is especially large. Accordingly, caution is warranted before generalizing the Cluster~2 estimate uniformly across all single-earner households.

\section{Conclusion and Future Works}\label{sec:conc}

This paper introduced CaSPECT, a causal spectral clustering framework that integrates causal discovery, causal effect estimation, and directed spectral embedding to identify causally homogeneous subgroups from observational data. By combining bootstrap-stabilized PC discovery with DirectLiNGAM through the proposed Orientation Validation Score (OVS), CaSPECT constructs a robust directed causal graph whose edges are weighted by backdoor-identified average causal effects and embedded using Chung's directed Laplacian. We establish consistency of the complete pipeline, including graph recovery, causal edge-weight estimation, and spectral embedding, under standard causal and semiparametric assumptions. Simulation studies demonstrated robustness to nonlinear causal mechanisms and mild violations of causal sufficiency, while the ablation study confirmed the contributions of OVS, bootstrap stability weighting, and edge contraction to reliable structural recovery. Applications to the LaLonde, IHDP, and 401(k) datasets showed that the proposed framework recovers interpretable causal subpopulations and reveals treatment-effect heterogeneity. These results establish a principled connection between causal inference and directed spectral graph theory, providing a new paradigm for subgroup discovery based on causal propagation pathways rather than covariate similarity. Future work will focus on extending CaSPECT to settings involving latent confounding, nonlinear causal discovery, uncertainty-aware spectral embeddings, temporal and dynamic causal graphs, and scalable, high-dimensional causal representations via graph neural networks, simplicial complexes, and topological learning frameworks.

\bibliography{references}


\begin{thebibliography}{30}
\ifx \bisbn   \undefined \def \bisbn  #1{ISBN #1}\fi
\ifx \binits  \undefined \def \binits#1{#1}\fi
\ifx \bauthor  \undefined \def \bauthor#1{#1}\fi
\ifx \batitle  \undefined \def \batitle#1{#1}\fi
\ifx \bjtitle  \undefined \def \bjtitle#1{#1}\fi
\ifx \bvolume  \undefined \def \bvolume#1{\textbf{#1}}\fi
\ifx \byear  \undefined \def \byear#1{#1}\fi
\ifx \bissue  \undefined \def \bissue#1{#1}\fi
\ifx \bfpage  \undefined \def \bfpage#1{#1}\fi
\ifx \blpage  \undefined \def \blpage #1{#1}\fi
\ifx \burl  \undefined \def \burl#1{\textsf{#1}}\fi
\ifx \doiurl  \undefined \def \doiurl#1{\url{https://doi.org/#1}}\fi
\ifx \betal  \undefined \def \betal{\textit{et al.}}\fi
\ifx \binstitute  \undefined \def \binstitute#1{#1}\fi
\ifx \binstitutionaled  \undefined \def \binstitutionaled#1{#1}\fi
\ifx \bctitle  \undefined \def \bctitle#1{#1}\fi
\ifx \beditor  \undefined \def \beditor#1{#1}\fi
\ifx \bpublisher  \undefined \def \bpublisher#1{#1}\fi
\ifx \bbtitle  \undefined \def \bbtitle#1{#1}\fi
\ifx \bedition  \undefined \def \bedition#1{#1}\fi
\ifx \bseriesno  \undefined \def \bseriesno#1{#1}\fi
\ifx \blocation  \undefined \def \blocation#1{#1}\fi
\ifx \bsertitle  \undefined \def \bsertitle#1{#1}\fi
\ifx \bsnm \undefined \def \bsnm#1{#1}\fi
\ifx \bsuffix \undefined \def \bsuffix#1{#1}\fi
\ifx \bparticle \undefined \def \bparticle#1{#1}\fi
\ifx \barticle \undefined \def \barticle#1{#1}\fi
\bibcommenthead
\ifx \bconfdate \undefined \def \bconfdate #1{#1}\fi
\ifx \botherref \undefined \def \botherref #1{#1}\fi
\ifx \url \undefined \def \url#1{\textsf{#1}}\fi
\ifx \bchapter \undefined \def \bchapter#1{#1}\fi
\ifx \bbook \undefined \def \bbook#1{#1}\fi
\ifx \bcomment \undefined \def \bcomment#1{#1}\fi
\ifx \oauthor \undefined \def \oauthor#1{#1}\fi
\ifx \citeauthoryear \undefined \def \citeauthoryear#1{#1}\fi
\ifx \endbibitem  \undefined \def \endbibitem {}\fi
\ifx \bconflocation  \undefined \def \bconflocation#1{#1}\fi
\ifx \arxivurl  \undefined \def \arxivurl#1{\textsf{#1}}\fi
\csname PreBibitemsHook\endcsname

\bibitem[\protect\citeauthoryear{Kim et~al.}{2024a}]{hierarchical}
\begin{barticle}
\bauthor{\bsnm{Kim}, \binits{K.}},
\bauthor{\bsnm{Kim}, \binits{J.}},
\bauthor{\bsnm{Wasserman}, \binits{L.}},
\bauthor{\bsnm{Kennedy}, \binits{E.}}:
\batitle{Hierarchical and density-based causal clustering}.
\bjtitle{Advances in Neural Information Processing Systems}
\bvolume{37},
\bfpage{30363}--\blpage{30393}
(\byear{2024})
\end{barticle}
\endbibitem

\bibitem[\protect\citeauthoryear{Kim et~al.}{2024b}]{causalkmeans}
\begin{botherref}
\oauthor{\bsnm{Kim}, \binits{K.}},
\oauthor{\bsnm{Kim}, \binits{J.}},
\oauthor{\bsnm{Kennedy}, \binits{E.H.}}:
Causal k-means clustering.
arXiv preprint arXiv:2405.03083
(2024)
\end{botherref}
\endbibitem

\bibitem[\protect\citeauthoryear{Von~Luxburg}{2007}]{spectral}
\begin{barticle}
\bauthor{\bsnm{Von~Luxburg}, \binits{U.}}:
\batitle{A tutorial on spectral clustering}.
\bjtitle{Statistics and computing}
\bvolume{17}(\bissue{4}),
\bfpage{395}--\blpage{416}
(\byear{2007})
\end{barticle}
\endbibitem

\bibitem[\protect\citeauthoryear{Hartigan and Wong}{1979}]{kmeans}
\begin{barticle}
\bauthor{\bsnm{Hartigan}, \binits{J.A.}},
\bauthor{\bsnm{Wong}, \binits{M.A.}}:
\batitle{Algorithm as 136: A k-means clustering algorithm}.
\bjtitle{Journal of the royal statistical society. series c (applied statistics)}
\bvolume{28}(\bissue{1}),
\bfpage{100}--\blpage{108}
(\byear{1979})
\end{barticle}
\endbibitem

\bibitem[\protect\citeauthoryear{Spirtes et~al.}{1993}]{pc_algo}
\begin{bchapter}
\bauthor{\bsnm{Spirtes}, \binits{P.}},
\bauthor{\bsnm{Glymour}, \binits{C.}},
\bauthor{\bsnm{Scheines}, \binits{R.}}:
\bctitle{Discovery algorithms for causally sufficient structures}.
In: \bbtitle{Causation, Prediction, and Search},
pp. \bfpage{103}--\blpage{162}.
\bpublisher{Springer}, \blocation{???}
(\byear{1993})
\end{bchapter}
\endbibitem

\bibitem[\protect\citeauthoryear{Neyman}{1923}]{Neyman:1923}
\begin{barticle}
\bauthor{\bsnm{Neyman}, \binits{J.}}:
\batitle{On the application of probability theory to agricultural experiments: Essay on principles (translated in statist. sci., 1990, 5, 465–472)}.
\bjtitle{Roczniki Nauk Rolniczych}
\bvolume{10},
\bfpage{1}--\blpage{51}
(\byear{1923})
\end{barticle}
\endbibitem

\bibitem[\protect\citeauthoryear{Rubin}{1974}]{Rubin:1974}
\begin{barticle}
\bauthor{\bsnm{Rubin}, \binits{D.B.}}:
\batitle{Estimating causal effects of treatments in randomized and nonrandomized studies}.
\bjtitle{J. Educ. Psychol.}
\bvolume{66},
\bfpage{688}--\blpage{701}
(\byear{1974})
\end{barticle}
\endbibitem

\bibitem[\protect\citeauthoryear{Spirtes et~al.}{2000}]{spirtes2000causation}
\begin{bbook}
\bauthor{\bsnm{Spirtes}, \binits{P.}},
\bauthor{\bsnm{Glymour}, \binits{C.N.}},
\bauthor{\bsnm{Scheines}, \binits{R.}}:
\bbtitle{Causation, Prediction, and Search}.
\bpublisher{MIT press}, \blocation{???}
(\byear{2000})
\end{bbook}
\endbibitem

\bibitem[\protect\citeauthoryear{Shimizu et~al.}{2006}]{shimizu2006linear}
\begin{botherref}
\oauthor{\bsnm{Shimizu}, \binits{S.}},
\oauthor{\bsnm{Hoyer}, \binits{P.O.}},
\oauthor{\bsnm{Hyv{\"a}rinen}, \binits{A.}},
\oauthor{\bsnm{Kerminen}, \binits{A.}},
\oauthor{\bsnm{Jordan}, \binits{M.}}:
A linear non-gaussian acyclic model for causal discovery.
Journal of Machine Learning Research
\textbf{7}(10)
(2006)
\end{botherref}
\endbibitem

\bibitem[\protect\citeauthoryear{Chernozhukov et~al.}{2018}]{chernozhukov2018double}
\begin{botherref}
\oauthor{\bsnm{Chernozhukov}, \binits{V.}},
\oauthor{\bsnm{Chetverikov}, \binits{D.}},
\oauthor{\bsnm{Demirer}, \binits{M.}},
\oauthor{\bsnm{Duflo}, \binits{E.}},
\oauthor{\bsnm{Hansen}, \binits{C.}},
\oauthor{\bsnm{Newey}, \binits{W.}},
\oauthor{\bsnm{Robins}, \binits{J.}}:
Double/debiased machine learning for treatment and structural parameters.
Oxford University Press Oxford, UK
(2018)
\end{botherref}
\endbibitem

\bibitem[\protect\citeauthoryear{McDonald}{2002}]{pearl}
\begin{barticle}
\bauthor{\bsnm{McDonald}, \binits{R.P.}}:
\batitle{Judea pearl. causality: Models, reasoning, and inference. cambridge: Cambridge university press. 384 pp., 2000, isbn 0521773628.}
\bjtitle{Psychometrika}
\bvolume{67}(\bissue{2}),
\bfpage{321}--\blpage{322}
(\byear{2002})
\end{barticle}
\endbibitem

\bibitem[\protect\citeauthoryear{Meek}{2013}]{meek2013causal}
\begin{botherref}
\oauthor{\bsnm{Meek}, \binits{C.}}:
Causal inference and causal explanation with background knowledge.
arXiv preprint arXiv:1302.4972
(2013)
\end{botherref}
\endbibitem

\bibitem[\protect\citeauthoryear{Meinshausen and B{\"u}hlmann}{2010}]{meinshausen2010stability}
\begin{barticle}
\bauthor{\bsnm{Meinshausen}, \binits{N.}},
\bauthor{\bsnm{B{\"u}hlmann}, \binits{P.}}:
\batitle{Stability selection}.
\bjtitle{Journal of the Royal Statistical Society Series B: Statistical Methodology}
\bvolume{72}(\bissue{4}),
\bfpage{417}--\blpage{473}
(\byear{2010})
\end{barticle}
\endbibitem

\bibitem[\protect\citeauthoryear{Shimizu et~al.}{2011}]{shimizu2011directlingam}
\begin{barticle}
\bauthor{\bsnm{Shimizu}, \binits{S.}},
\bauthor{\bsnm{Inazumi}, \binits{T.}},
\bauthor{\bsnm{Sogawa}, \binits{Y.}},
\bauthor{\bsnm{Hyvarinen}, \binits{A.}},
\bauthor{\bsnm{Kawahara}, \binits{Y.}},
\bauthor{\bsnm{Washio}, \binits{T.}},
\bauthor{\bsnm{Hoyer}, \binits{P.O.}},
\bauthor{\bsnm{Bollen}, \binits{K.}},
\bauthor{\bsnm{Hoyer}, \binits{P.}}:
\batitle{Directlingam: A direct method for learning a linear non-gaussian structural equation model}.
\bjtitle{Journal of Machine Learning Research-JMLR}
\bvolume{12}(\bissue{Apr}),
\bfpage{1225}--\blpage{1248}
(\byear{2011})
\end{barticle}
\endbibitem

\bibitem[\protect\citeauthoryear{Thadewald and B{\"u}ning}{2007}]{jarque}
\begin{barticle}
\bauthor{\bsnm{Thadewald}, \binits{T.}},
\bauthor{\bsnm{B{\"u}ning}, \binits{H.}}:
\batitle{Jarque--bera test and its competitors for testing normality--a power comparison}.
\bjtitle{Journal of applied statistics}
\bvolume{34}(\bissue{1}),
\bfpage{87}--\blpage{105}
(\byear{2007})
\end{barticle}
\endbibitem

\bibitem[\protect\citeauthoryear{Zheng et~al.}{2018}]{zheng2018dags}
\begin{botherref}
\oauthor{\bsnm{Zheng}, \binits{X.}},
\oauthor{\bsnm{Aragam}, \binits{B.}},
\oauthor{\bsnm{Ravikumar}, \binits{P.K.}},
\oauthor{\bsnm{Xing}, \binits{E.P.}}:
Dags with no tears: Continuous optimization for structure learning.
Advances in neural information processing systems
\textbf{31}
(2018)
\end{botherref}
\endbibitem

\bibitem[\protect\citeauthoryear{Hastie and Tibshirani}{1986}]{gam}
\begin{barticle}
\bauthor{\bsnm{Hastie}, \binits{T.}},
\bauthor{\bsnm{Tibshirani}, \binits{R.}}:
\batitle{Generalized additive models}.
\bjtitle{Statistical science}
\bvolume{1}(\bissue{3}),
\bfpage{297}--\blpage{310}
(\byear{1986})
\end{barticle}
\endbibitem

\bibitem[\protect\citeauthoryear{Breiman}{2001}]{random}
\begin{barticle}
\bauthor{\bsnm{Breiman}, \binits{L.}}:
\batitle{Random forests}.
\bjtitle{Machine learning}
\bvolume{45}(\bissue{1}),
\bfpage{5}--\blpage{32}
(\byear{2001})
\end{barticle}
\endbibitem

\bibitem[\protect\citeauthoryear{Page et~al.}{1999}]{page1999pagerank}
\begin{botherref}
\oauthor{\bsnm{Page}, \binits{L.}},
\oauthor{\bsnm{Brin}, \binits{S.}},
\oauthor{\bsnm{Motwani}, \binits{R.}},
\oauthor{\bsnm{Winograd}, \binits{T.}}:
The pagerank citation ranking: Bringing order to the web.
Technical report,
Stanford infolab
(1999)
\end{botherref}
\endbibitem

\bibitem[\protect\citeauthoryear{Chung}{2005}]{chung2005laplacians}
\begin{barticle}
\bauthor{\bsnm{Chung}, \binits{F.}}:
\batitle{Laplacians and the cheeger inequality for directed graphs}.
\bjtitle{Annals of Combinatorics}
\bvolume{9}(\bissue{1}),
\bfpage{1}--\blpage{19}
(\byear{2005})
\end{barticle}
\endbibitem

\bibitem[\protect\citeauthoryear{Chung}{1997}]{chung1997spectral}
\begin{bbook}
\bauthor{\bsnm{Chung}, \binits{F.R.}}:
\bbtitle{Spectral Graph Theory}
vol. \bseriesno{92}.
\bpublisher{American Mathematical Soc.}, \blocation{???}
(\byear{1997})
\end{bbook}
\endbibitem

\bibitem[\protect\citeauthoryear{Davis and Kahan}{1970}]{davis1970rotation}
\begin{barticle}
\bauthor{\bsnm{Davis}, \binits{C.}},
\bauthor{\bsnm{Kahan}, \binits{W.M.}}:
\batitle{The rotation of eigenvectors by a perturbation. iii}.
\bjtitle{SIAM Journal on Numerical Analysis}
\bvolume{7}(\bissue{1}),
\bfpage{1}--\blpage{46}
(\byear{1970})
\end{barticle}
\endbibitem

\bibitem[\protect\citeauthoryear{Peter and ROUSSEEUW}{1987}]{peter1987graphical}
\begin{barticle}
\bauthor{\bsnm{Peter}, \binits{J.}},
\bauthor{\bsnm{ROUSSEEUW}, \binits{S.}}:
\batitle{A graphical aid to the interpretation and validation of cluster analysis}.
\bjtitle{J. Comput. Appl. Math}
\bvolume{20},
\bfpage{53}--\blpage{65}
(\byear{1987})
\end{barticle}
\endbibitem

\bibitem[\protect\citeauthoryear{Tibshirani et~al.}{2001}]{tibshirani2001estimating}
\begin{barticle}
\bauthor{\bsnm{Tibshirani}, \binits{R.}},
\bauthor{\bsnm{Walther}, \binits{G.}},
\bauthor{\bsnm{Hastie}, \binits{T.}}:
\batitle{Estimating the number of clusters in a data set via the gap statistic}.
\bjtitle{Journal of the royal statistical society: series b (statistical methodology)}
\bvolume{63}(\bissue{2}),
\bfpage{411}--\blpage{423}
(\byear{2001})
\end{barticle}
\endbibitem

\bibitem[\protect\citeauthoryear{Santos and Embrechts}{2009}]{ari}
\begin{bchapter}
\bauthor{\bsnm{Santos}, \binits{J.M.}},
\bauthor{\bsnm{Embrechts}, \binits{M.}}:
\bctitle{On the use of the adjusted rand index as a metric for evaluating supervised classification}.
In: \bbtitle{International Conference on Artificial Neural Networks},
pp. \bfpage{175}--\blpage{184}
(\byear{2009}).
\bcomment{Springer}
\end{bchapter}
\endbibitem

\bibitem[\protect\citeauthoryear{LaLonde}{1986}]{lalonde1986}
\begin{botherref}
\oauthor{\bsnm{LaLonde}, \binits{R.J.}}:
Evaluating the econometric evaluations of training programs with experimental data.
The American economic review,
604--620
(1986)
\end{botherref}
\endbibitem

\bibitem[\protect\citeauthoryear{Louizos et~al.}{2017}]{louizos2017causal}
\begin{botherref}
\oauthor{\bsnm{Louizos}, \binits{C.}},
\oauthor{\bsnm{Shalit}, \binits{U.}},
\oauthor{\bsnm{Mooij}, \binits{J.M.}},
\oauthor{\bsnm{Sontag}, \binits{D.}},
\oauthor{\bsnm{Zemel}, \binits{R.}},
\oauthor{\bsnm{Welling}, \binits{M.}}:
Causal effect inference with deep latent-variable models.
Advances in neural information processing systems
\textbf{30}
(2017)
\end{botherref}
\endbibitem

\bibitem[\protect\citeauthoryear{Poterba et~al.}{1995}]{poterba1995401}
\begin{barticle}
\bauthor{\bsnm{Poterba}, \binits{J.M.}},
\bauthor{\bsnm{Venti}, \binits{S.F.}},
\bauthor{\bsnm{Wise}, \binits{D.A.}}:
\batitle{Do 401 (k) contributions crowd out other personal saving?}
\bjtitle{Journal of Public Economics}
\bvolume{58}(\bissue{1}),
\bfpage{1}--\blpage{32}
(\byear{1995})
\end{barticle}
\endbibitem

\bibitem[\protect\citeauthoryear{Slutsky}{1925}]{slutsky}
\begin{barticle}
\bauthor{\bsnm{Slutsky}, \binits{E.}}:
\batitle{{\"U}ber stochastische asymptoten und grenzwerte}.
\bjtitle{Metron}
\bvolume{5}(\bissue{3}),
\bfpage{3}--\blpage{89}
(\byear{1925})
\end{barticle}
\endbibitem

\bibitem[\protect\citeauthoryear{Kim and Pollard}{1990}]{cmt}
\begin{botherref}
\oauthor{\bsnm{Kim}, \binits{J.}},
\oauthor{\bsnm{Pollard}, \binits{D.}}:
Cube root asymptotics.
The Annals of Statistics,
191--219
(1990)
\end{botherref}
\endbibitem

\end{thebibliography}

\appendix

\section{Appendix}
\subsection{Proofs of Theorems and Propositions.} 
\vspace*{12pt}

\subsubsection{Proof of the Proposition \ref{prop:ovs}} \label{prop:proof}

\begin{proof}
We partition the true edges of the underlying DAG into two disjoint sets: edges structurally identifiable by the constraint-based skeleton and edges trapped within a Markov Equivalence Class (MEC). By the consistency of the PC algorithm \citep{pc_algo}, for identifiable edges, the bootstrap orientation frequencies converge almost surely such that $\delta^{\text{PC}}_{uv} \to \pm 1$. Conversely, for edges within the MEC, the constraint-based orientation is uninformative, yielding $\delta^{\text{PC}}_{uv} \to c$, where $c \in (-1, 1)$ (typically $c \to 0$).Concurrently, by the consistency of DirectLiNGAM under non-Gaussian structural equations in Assumption \ref{ass:B3} \citep{shimizu2011directlingam}, $\delta^L_{uv} \to \pm 1$ for all true edges.We evaluate the asymptotic behavior of $\text{OVS}_{uv} = w_{\text{PC}} \delta^{\text{PC}}_{uv} + w_L \delta^L_{uv}$ under these two cases. Since, in the asymptotic limit $f_{uv} \rightarrow 1$ by stability selection \citep{meinshausen2010stability}.

\noindent
Case 1 (Identifiable Edges): Both the constraint-based and continuous optimization components converge to the correct sign. The score converges to $\text{OVS}_{uv} \to \pm(w_{\text{PC}} + w_L) = \pm 1$. Since $1 > \gamma$, the edge is consistently oriented.

\noindent
Case 2 (MEC Edges): The PC component is ambiguous ($\delta^{\text{PC}}_{uv} \to c$), but the LiNGAM component converges to the true causal sign. The score converges to $w_{\text{PC}}(c) \pm w_L$. In the worst-case scenario, where the PC algorithm provides zero directional information ($c=0$), the score converges exactly to $\pm w_L$. By our explicit pipeline, $w_L > \gamma$. Therefore, $|\text{OVS}_{uv}| > \gamma$ is strictly guaranteed even for MEC edges. Because the continuous optimization component governs the sign, the OVS reliably escapes the ambiguity threshold and assigns the correct orientation, ensuring consistent asymptotic recovery of the true DAG.
\end{proof}

\subsubsection{Proof of the Theorem \ref{thm:ate}} \label{thm:proof}

\begin{proof}
The proof has three parts, corresponding to Track~OLS, Track~DML, and the stability weighting argument.
 
\medskip
\noindent\textbf{Part I: Linear Edges corresponding to OLS.}
 
\medskip
The RESET test does not reject linearity, so we work with the structural
equation
\[
  v = \beta_{uv}\,u + \bm{\beta}_X^\top\X_{uv} + \varepsilon_v,
  \qquad
  \bbE[\varepsilon_v \mid u,\X_{uv}] = 0.
\]
The error $\varepsilon_v$ captures all variation in $v$ not linearly
explained by $u$ and $\X_{uv}$.
The conditional moment restriction $\bbE[\varepsilon_v \mid u, \X_{uv}] = 0$
is the structural assumption that makes the coefficient $\beta_{uv}$
interpretable as a causal quantity: it says the structural noise is
orthogonal to both the treatment and the adjusting covariates, which is
precisely what separates a causal parameter from a predictive one.
 
Under consistency \citep{Neyman:1923, Rubin:1974}, the potential outcome of unit $i$ under the intervention $\mathrm{do}(u = a)$ equals the value the structural equation assigns when $u$ is set to $a$.
Substituting $u = a$ into the structural equation gives,
\[
  v(a) = \beta_{uv}\,a + \bm{\beta}_X^\top\X_{uv} + \varepsilon_v.
\]
This step rests on consistency in the sense of
\citep{Neyman:1923} and \citep{Rubin:1974}: the potential outcome $v(a)$
equals the observed outcome whenever unit $i$ actually receives treatment
$u = a$.
SUTVA requires that there is no interference between units and that the treatment is well-defined. The covariates $\X_{uv}$ and the error $\varepsilon_v$ are held at their
natural values; the intervention changes $u$ alone, not the background variables or the noise.
This is the structural reading of the equation: it describes a mechanism, not a correlation.

\noindent
Taking unconditional expectations on both sides gives
\[
  \bbE[v(a)] = \beta_{uv}\,a + \bm{\beta}_X^\top\bbE[\X_{uv}].
\]
The linearity of the structural equation means $a$ enters through
$\beta_{uv}$ alone.
The term $\bm{\beta}_X^\top\bbE[\X_{uv}]$ is constant in $a$ because
$\X_{uv}$ consists of pre-treatment covariates that the intervention on $u$ does not affect; this is guaranteed by the backdoor criterion, which restricts $\X_{uv}$ to non-descendants of $u$. Contrasting $a = 1$ and $a = 0$ therefore gives
\[
  \tau_{u \to v}
  = \bbE[v(1) - v(0)]
  = \beta_{uv}.
\]
The ACE equals the structural coefficient $\beta_{uv}$.
The backdoor criterion blocks every confounding path from $u$ to $v$
through $\X_{uv}$, so the structural coefficient and the regression
coefficient in the adjusted regression coincide at the population level.
 
It remains to show that the population OLS coefficient on $u$ in the
regression of $v$ on $(u, \X_{uv})$ equals $\beta_{uv}$ and not some
other quantity.
The population OLS minimises
$\bbE[(v - b_u u - \mathbf{b}_X^\top\X_{uv})^2]$ over
$(b_u, \mathbf{b}_X)$.
The first-order condition with respect to $b_u$ requires
\[
  \bbE\!\left[
    \bigl(v - \beta_{uv}\,u - \bm{\beta}_X^\top\X_{uv}\bigr)\cdot u
  \right] = 0.
\]
Substituting the structural equation for $v$ reduces this to
$\bbE[\varepsilon_v \cdot u] = 0$.
This is not assumed --- it is derived.
By the law of iterated expectations,
\[
  \bbE[\varepsilon_v \cdot u]
  = \bbE\!\left[\bbE[\varepsilon_v \cdot u \mid u,\X_{uv}]\right]
  = \bbE\!\left[u \cdot \bbE[\varepsilon_v \mid u,\X_{uv}]\right]
  = \bbE[u \cdot 0]
  = 0.
\]
The second equality pulls $u$ outside the inner expectation because $u$
is measurable with respect to the conditioning sigma-field
$(u, \X_{uv})$.
The third equality applies the conditional moment restriction.
The same argument applies to each component of $\X_{uv}$, confirming
that $(\beta_{uv}, \bm{\beta}_X)$ solves the full normal equations.
 
This step is non-trivial in the observational setting.
Without adjustment for $\X_{uv}$, omitted confounders would enter the
residual and create a non-zero correlation with $u$, so the marginal
regression of $v$ on $u$ would not recover $\beta_{uv}$.
The backdoor adjustment set removes exactly this confounding, and the
above argument shows that once $\X_{uv}$ is included the residual is
orthogonal to $u$ in the population.

The sample OLS estimator $\hat{\beta}_{uv}$ solves the sample analog of the normal equations.
Under standard OLS regularity conditions --- finite second moments of
$(v, u, \X_{uv})$ and a non-singular population second moment matrix, the sample normal equations converge to their population counterparts uniformly in probability by the law of large numbers. By Slutsky's theorem \citep{slutsky} and the continuous mapping theorem \citep{cmt},
$\hat{\beta}_{uv} \xrightarrow{a.s.} \beta_{uv} = \tau_{u \to v}$,
and therefore $|\hat{\beta}_{uv}| \xrightarrow{a.s.} |\beta_{uv}|$.
 
\medskip
\noindent\textbf{Part II: Nonlinear Edges corresponding to Double Machine Learning.}

The RESET test rejects linearity, so the linear SEM assumption is not credible for this edge. We instead work with the partially linear model
\[
  v = \tau_{u \to v}\cdot u + g(\X_{uv}) + \xi_v,
  \qquad
  \bbE[\xi_v \mid u,\X_{uv}] = 0,
\]
where $g\colon\bbR^{|\X_{uv}|} \to \bbR$ is an unknown smooth function
capturing the nonlinear effect of the covariates on $v$, and the
treatment effect $\tau_{u \to v}$ is a constant structural parameter.
All nonlinearity is absorbed into $g$; the parameter of interest enters linearly. If one were to regress $v$ on $u$ and $\X_{uv}$ using OLS in this
nonlinear setting, the misspecification of $g$ as a linear function
would induce bias in the coefficient on $u$.
The OLS slope would absorb part of the nonlinear confounding effect of
$\X_{uv}$ on $v$, producing an estimate that does not equal
$\tau_{u \to v}$ at the population level.
This is the failure mode Track~B is designed to avoid.
 
The DML approach recovers $\tau_{u \to v}$ by projecting out the influence of $\X_{uv}$ from both $v$ and $u$ before regressing. Define the population conditional mean functions
\[
  m(\X_{uv}) = \bbE[v \mid \X_{uv}],
  \qquad
  \ell(\X_{uv}) = \bbE[u \mid \X_{uv}],
\]
and the corresponding population residuals
\[
  \tilde{V} = v - m(\X_{uv}),
  \qquad
  \tilde{U} = u - \ell(\X_{uv}).
\]
Substituting the partially linear model into the expression for
$\tilde{V}$ gives,
\[
  \tilde{V}
  = \tau_{u \to v}(u - \ell(\X_{uv}))
    + (g(\X_{uv}) - m(\X_{uv}))
    + \xi_v.
\]
Observe that $m(\X_{uv}) = \bbE[v \mid \X_{uv}]
= \tau_{u \to v}\,\ell(\X_{uv}) + g(\X_{uv})$,
so the term $g(\X_{uv}) - m(\X_{uv})
= -\tau_{u \to v}\,\ell(\X_{uv})$ cancels with the residualized
treatment contribution, leaving
\[
  \tilde{V} = \tau_{u \to v}\,\tilde{U} + \xi_v.
\]
The unknown nonlinear function $g$ has been eliminated entirely.
The population regression of $\tilde{V}$ on $\tilde{U}$ therefore
recovers $\tau_{u \to v}$ without bias from the nonlinearity in $g$.
 
For the partialled-out regression to be unbiased, we need
$\bbE[\xi_v\,\tilde{U}] = 0$.
Since $\tilde{U} = u - \ell(\X_{uv})$ is measurable with respect to
$(u, \X_{uv})$, the law of iterated expectations gives
\[
  \bbE[\xi_v\,\tilde{U}]
  = \bbE\!\left[\tilde{U}\cdot\bbE[\xi_v \mid u,\X_{uv}]\right]
  = \bbE[\tilde{U}\cdot 0]
  = 0,
\]
where the second equality uses $\bbE[\xi_v \mid u,\X_{uv}] = 0$, which
follows from conditional ignorability given $\X_{uv}$.
This is the Neyman orthogonality condition
\citep{chernozhukov2018double}: it ensures that small perturbations in
the nuisance functions $m$ and $\ell$ do not produce first-order bias in
the estimate of $\tau_{u \to v}$.
 
In practice, $m$ and $\ell$ are unknown and must be estimated from the
same data used to estimate $\tau_{u \to v}$.
Using the same observations for both nuisance estimation and treatment
effect estimation would induce in-sample overfitting, creating a
correlation between the estimated residuals and the true residuals that
biases the final estimator.
Cross-fitting prevents this by partitioning the sample into $K$ folds
and estimating nuisances on the complement of each fold before computing
residuals on that fold.
The cross-fitted residuals $\tilde{U}_i = u_i - \hat{\ell}^{(-k)}(\X_{uv,i})$
and $\tilde{V}_i = v_i - \hat{m}^{(-k)}(\X_{uv,i})$ are therefore
constructed from estimators that have not seen observation $i$, so the
product $\tilde{U}_i\tilde{V}_i$ behaves as if the nuisance functions
were known at the population level, up to first order.

\noindent
The DML estimator can be written as
\[
  \hat{\tau}_{\mathrm{DML}}
  = \tau_{u \to v}
    + \frac{\sum_i \tilde{U}_i\,\xi_i}{\sum_i \tilde{U}_i^2}
    + R_n,
\]
where $R_n$ collects terms arising from nuisance estimation error.
The second term converges to zero in probability by the law of large
numbers and the orthogonality condition $\bbE[\xi_v\,\tilde{U}] = 0$.
Under Neyman orthogonality, the remainder $R_n$ is second-order in the
estimation errors of $\hat{m}$ and $\hat{\ell}$.
\citep{chernozhukov2018double} (Theorem~3.1) shows that if the nuisance estimators converge at a rate faster than $n^{-1/4}$ in $L_2$ norm, a condition satisfied by GAM and random forest estimators under standard smoothness assumptions, then $R_n = \mathcal{O}(n^{-1/2})$.
It follows that
\[
  \sqrt{n}(\hat{\tau}_{\mathrm{DML}} - \tau_{u \to v})
  \xrightarrow{d} \mathcal{N}(0,\sigma^2)
\]
for a finite variance $\sigma^2$, and in particular
$\hat{\tau}_{\mathrm{DML}} \xrightarrow{a.s.} \tau_{u \to v}$,
so $|\hat{\tau}_{\mathrm{DML}}| \xrightarrow{a.s.} |\tau_{u \to v}|$.
 
\medskip
\noindent\textbf{Part III: Bootstrap Stability Weighting.}
 
The inclusion frequency $f_{uv} = B^{-1}\sum_{b=1}^B
\mathbf{1}[(u,v) \in \text{skeleton of } G_\mathrm{PC}^{(b)}]$
counts the fraction of bootstrap resamples in which edge $(u,v)$
appears.
Under faithfulness and causal sufficiency, PC Algorithm consistency
\citep{spirtes2000causation} implies that the probability of including
a true edge in the skeleton converges to~1 as $n \to \infty$.
Applying the law of large numbers across bootstrap resamples therefore gives $f_{uv} \xrightarrow{a.s.} \mathbf{1}_{(u,v) \in E_0}$. In particular, for every true edge $f_{uv} \to 1$ almost surely, and
for every false edge $f_{uv} \to 0$ almost surely. For a true edge $(u \to v) \in E_0$, Parts~I or~II give
$A_{uv} \xrightarrow{a.s.} |\tau_{u \to v}^0|$.
Since $f_{uv} \to 1$ almost surely, the continuous mapping theorem \citep{cmt} applied to the product of two almost-surely convergent sequences gives
\[
  A_{uv}^{\mathrm{stab}}
  = f_{uv} \cdot A_{uv}
  \;\xrightarrow{a.s.}\;
  1 \cdot |\tau_{u \to v}^0|
  = |\tau_{u \to v}^0|.
\]
For a false edge $(u,v) \notin E_0$, faithfulness and causal sufficiency
ensure $f_{uv} \to 0$ almost surely.
Under the finite-moment conditions of Tracks~OLS and~DML, $A_{uv}$ remains
bounded, so
\[
  A_{uv}^{\mathrm{stab}} = f_{uv} \cdot A_{uv} \to 0 \cdot A_{uv} = 0
  \quad \text{almost surely.}
\]
Spurious edges that appear in bootstrap resamples due to finite-sample variations are therefore suppressed in the adjacency matrix without any analyst-specified hard threshold.
The stability weighting acts as a data-driven regulariser.
 
\medskip
Combining Parts~I, II, and~III, $A_{uv}^{\mathrm{stab}}$ converges
almost surely to $|\tau_{u \to v}^0|$ for every directed edge
$(u \to v) \in E^*$ as $n \to \infty$.
\end{proof}

\subsection{Consistency of CaSPECT} \label{thm:consistency}
We now show the consistency of our proposed method, CaSPECT. 
\begin{theorem} \label{thm:main}
Under Assumptions~\ref{ass:B1}--\ref{ass:B5}, as $n \to \infty$:
\begin{enumerate}
  \item[\emph{(a)}] $\calS \xrightarrow{a.s.} E_0$, by PC consistency
        \citep{spirtes2000causation} and convergence of bootstrap
        frequencies.
  \item[\emph{(b)}] OVS orientations converge to those of $G_0$ a.s.,
        by Proposition~\ref{prop:ovs}.
  \item[\emph{(c)}] $A_{uv}^{\mathrm{stab}} \xrightarrow{a.s.}
        \abs{\tau_{u \to v}^0}$, by OLS/DML consistency (Theorem~\ref{thm:ate}) and $f_{uv} \to 1$ from part (a).
  \item[\emph{(d)}] $L \xrightarrow{a.s.} L_0$.
  \item[\emph{(e)}] $\xt$ converges to the true causal spectral embedding up to rotation within the eigenspace.
\end{enumerate}
\end{theorem}

\begin{proof}

\textbf{(a)} Let,
\[
\mathcal{S}
=
\left\{
(u,v): f_{uv}\ge \theta
\right\}
\]
is built from bootstrap inclusion frequencies
\[
f_{uv}
=
\frac{1}{B}
\sum_{b=1}^{B}
\mathbf{1}
\!\left[
(u,v)\in
\text{skeleton of }G_{\mathrm{PC}}^{(b)}
\right].
\]

Under Assumptions~\ref{ass:B1} and \ref{ass:B2}, the PC algorithm consistently recovers the true CPDAG when the significance level satisfies
\[
\alpha_n \to 0,
\qquad
n\alpha_n \to \infty .
\]
Consequently,
\[
\mathbb{P}
\!\left[
(u,v)\in
\text{skeleton of }
G_{\mathrm{PC}}^{(b)}
\right]
\longrightarrow
\mathbf{1}_{(u,v)\in E_0}
\]
for every node pair $(u,v)$.

For fixed $n$, the frequencies $f_{uv}$ are averages of Bernoulli variables across bootstrap resamples. By the Strong Law of Large Numbers,
\[
f_{uv}
\xrightarrow{a.s.}
\mathbb{P}
\!\left[
(u,v)\in
\text{skeleton of }
G_{\mathrm{PC}}^{(b)}
\right]
\qquad
(B\to\infty).
\]
Combining the bootstrap convergence with the consistency of PC gives
\[
f_{uv}
\xrightarrow{a.s.}
\mathbf{1}_{(u,v)\in E_0}
\qquad
(n\to\infty).
\]
Since $\theta\in(0,1)$ is fixed,
\[
\{f_{uv}\ge \theta\}
\xrightarrow{a.s.}
\{(u,v)\in E_0\},
\]
which implies,
\[
\mathcal{S}
\xrightarrow{a.s.}
E_0.
\]

\noindent
\textbf{(b)} This proof follows from Proposition \ref{prop:ovs}.

\noindent
\textbf{(c)} This proof follows from Theorem \ref{thm:ate}.
 
\noindent
\textbf{(d)} Define the sequence of events
\[
\mathcal{A}_n
=
\left\{
\hat{G}^{*}=G_0
\right\}.
\]
\noindent
By (a) and (b), the estimated stable causal graph converges to the true DAG, implying
\[
\mathbb{P}(\mathcal{A}_n)
\longrightarrow
1
\qquad
\text{as }
n \to \infty.
\]

On the event $\mathcal{A}_n$, Part~(c) establishes that the stability-weighted adjacency matrix converges almost surely to the population adjacency matrix of true average causal effect (ACE) magnitudes:
\[
\mathbf{A}^{\mathrm{stab}}
\xrightarrow{a.s.}
\mathbf{A}_0,
\]
where $\mathbf{A}_0$ denotes the true ACE-weighted adjacency matrix.
\noindent
The Chung Laplacian is obtained from
$\mathbf{A}^{\mathrm{stab}}$
through a sequence of continuous transformations. First, row-normalization together with the PageRank teleportation parameter $\alpha>0$ yields the transition matrix $P$.

\noindent
The corresponding stationary distribution
\[
\pi
=
(\pi_1,\ldots,\pi_p)^\top
\]
is defined as the unique left unit eigenvector satisfying
\[
\pi^\top P
=
\pi^\top,
\qquad
\sum_{i=1}^{p}\pi_i=1.
\]
\noindent
Because $\alpha>0$, the Markov chain induced by $P$ is irreducible and aperiodic, ensuring $\pi_i>0, i=1,\ldots,p.$

\noindent
The time-reversed transition matrix is then defined by,
\[
P^{*}_{ij}
=
\frac{\pi_j}{\pi_i}
P_{ji}.
\]
Since each $\pi_i$ is strictly positive, the mapping
\[
(P,\pi)
\mapsto
P^{*}
\]
is continuous.

\noindent
Finally, the Chung directed Laplacian is given by
\[
L
=
\mathbf{I}
-
\frac{1}{2}
\left(
P+P^{*}
\right),
\]
which is itself a continuous function of $P$ and $P^{*}$.

\noindent
Therefore, the entire construction
\[
\mathbf{A}^{\mathrm{stab}}
\longmapsto
P
\longmapsto
\pi
\longmapsto
P^{*}
\longmapsto
L
\]
is a composition of continuous maps. Since
\[
\mathbf{A}^{\mathrm{stab}}
\xrightarrow{a.s.}
\mathbf{A}_0,
\]
the Continuous Mapping Theorem \citep{cmt} implies that
\[
L
\xrightarrow{a.s.}
L_0,
\]
where $L_0$ denotes the Chung Laplacian associated with the true DAG $G_0$ and the true ACE-weighted adjacency matrix $\mathbf{A}_0$. Hence, the estimated Chung Laplacian converges almost surely to its population counterpart as the sample size tends to infinity.

\noindent
\textbf{(e)} Let,
\[
\delta_{K^{*}}
=
\lambda_{K^{*}+1}
-
\lambda_{K^{*}}
>
0
\]
denote the spectral gap of $L_0$.

\noindent
Define the perturbation matrix
\[
E
=
\hat L
-
L_0.
\]
For both OLS and DML,
\[
|\hat\theta-\theta|
=
\mathcal{O}(n^{-1/2}),
\]
where $\theta$ denotes the corresponding population ACE parameter.

\noindent
Therefore,
\[
\|{\mathbf A}^{\mathrm{stab}}
-
\mathbf A_0
\|_F
=
\mathcal{O}
\!\left(
\frac{|\mathbf V^{*}|}{\sqrt n}
\right).
\]
By continuity of the Laplacian construction,
\[
\|E\|_F
=
\|\hat L-L_0\|_F
=
\mathcal{O}
\!\left(
\frac{|\mathbf V^{*}|}{\sqrt n}
\right).
\]
Applying the Davis--Kahan \citep{davis1970rotation} theorem,
\[
\|
\sin\Theta
(
\mathbf V_{K^*},
\hat{\mathbf V}_{K^*}
)
\|_F
\le
\frac{\|E\|_F}{\delta_{K^{*}}}.
\]
Therefore, the embedding $\tilde{\mathbf X}=
\mathbf X^{*}\hat{\mathbf V}_{K^*}$ converges to the population embedding $\mathbf X^{*}\mathbf V_{K^*}$ up to an orthogonal rotation.

\noindent
Combining Parts~(a)--(e), the entire CaSPECT pipeline is almost surely consistent.
\end{proof}

\end{document}